\documentclass[12pt]{article}

\usepackage[a4paper,left=20mm,top=20mm,right=20mm,bottom=30mm]{geometry}
\usepackage{lastpage}
\usepackage{color}
\usepackage{amssymb, amsthm}
\usepackage{mathtools}
\usepackage{algorithmic,algorithm}

\usepackage{authblk}
\usepackage[numbers, sort&compress]{natbib}

\usepackage[colorlinks]{hyperref}
\hypersetup{
	citecolor=blue,
   linkcolor=red,
}
\usepackage[font=footnotesize,width=.85\textwidth,labelfont=bf]{caption}
\usepackage[mathcal]{euscript}
\usepackage{mathptmx}

\makeatletter	
\providecommand*{\cupdot}{%
  \mathbin{%
    \mathpalette\@cupdot{}%
  }%
}
\newcommand*{\@cupdot}[2]{%
  \ooalign{%
    $\m@th#1\cup$\cr
    \hidewidth$\m@th#1\cdot$\hidewidth
  }%
}
\providecommand*{\bigcupdot}{%
  \mathop{%
    \vphantom{\bigcup}%
    \mathpalette\@bigcupdot{}%
  }%
}
\newcommand*{\@bigcupdot}[2]{%
  \ooalign{%
    $\m@th#1\bigcup$\cr
    \sbox0{$#1\bigcup$}%
    \dimen@=\ht0 %
    \advance\dimen@ by -\dp0 %
    \sbox0{\scalebox{2}{$\m@th#1\cdot$}}%
    \advance\dimen@ by -\ht0 %
    \dimen@=.5\dimen@
    \hidewidth\raise\dimen@\box0\hidewidth
  }%
}
\newcommand{\raisemath}[1]{\mathpalette{\raisem@th{#1}}}
\newcommand{\raisem@th}[3]{\raisebox{#1}{$#2#3$}}
\providecommand*{\sqplus}{%
  \mathbin{%
    \mathpalette\@sqplus{}%
  }%
}
\newcommand*{\@sqplus}[2]{%
  \ooalign{%
    $\m@th#1\sqcup$\cr
    \hidewidth$\m@th#1\raisemath{1.4pt}{\scriptscriptstyle{+}}$\hidewidth
  }%
}
\makeatother

\DeclareMathOperator{\lca}{lca}
\newcommand{\DELTA}{\bigtriangleup}
\renewcommand{\P}{\ensuremath{\operatorname{\mathbb{P}}}}
\newcommand{\Pmax}{\ensuremath{\operatorname{\mathbb{P_{\max}}}}}
\DeclareMathOperator{\M}{MD}
\DeclareMathOperator{\MD}{MDs}
\newcommand{\MDT}{T_{\textrm{MDs}}}
\newcommand{\Fs}{\ensuremath{F^{\star}}}
\newcommand{\Gs}{\ensuremath{G^{\star}}}
\newcommand{\Ms}{\ensuremath{M^{\star}}}
\newcommand{\Ns}{\ensuremath{N^{\star}}}
\newcommand{\out}[1]{out$_{#1}$}
\newcommand{\merge}{\ensuremath{\sqplus}}
\newcommand{\mc}{\ensuremath{\mathcal}}
\newcommand{\rt}[1]{\ensuremath{\mathrm{#1}}}

\newtheorem{theorem}{Theorem}[section]
\newtheorem{lemma}[theorem]{Lemma}
\newtheorem{proposition}[theorem]{Proposition}
\newtheorem{definition}{Definition}[section]

\newtheorem{problem}{Problem} %[theorem]

\newtheorem{Obs}{Observation}

\providecommand{\keywords}[1]{\textbf{\textit{Keywords: }} #1}

\makeindex             % used for the subject index
\title{Cograph Editing: Merging Modules is equivalent to Editing $P_4$s}
\author[1]{Adrian Fritz}
\author[2]{Marc Hellmuth}
\author[3]{Peter F. Stadler} 
\author[4]{Nicolas Wieseke}

\affil[1]{Computational Biology of Infection Research, 
  Helmholtz Centre for Infection Research, Inhoffenstra{\ss}e 7, 
  D-38124 Braunschweig, Germany\\
	Email: \texttt{adrian.fritz@helmholtz-hzi.de}}

\affil[2]{Institute	 of Mathematics and Computer Science, University of Greifswald, Walther-
  Rathenau-Strasse 47, D-17487 Greifswald, Germany  \\ 	
	Email: \texttt{mhellmuth@mailbox.org}}

\affil[3]{Swarm Intelligence and Complex Systems Group, Department of
  Computer Science, Leipzig University, Augustusplatz 10, D-04109
  Leipzig, Germany  \\ 	
	Email: \texttt{wieseke@informatik.uni-leipzig.de}}

\affil[4]{Bioinformatics Group, Department of Computer Science, 
 Universit{\"a}t Leipzig, H{\"a}rtelstrasse 16-18, D-04107 Leipzig,
 Germany \\ 	
	Email: \texttt{studla@bioinf.uni-leipzig.de}}

\date{}

\begin{document}

\maketitle

\abstract{ 
  The modular decomposition of a graph $G=(V,E)$ does not contain prime
  modules if and only if $G$ is a cograph, that is, if no quadruple of
  vertices induces a simple connected path $P_4$. The cograph editing
  problem consists in inserting into and deleting from $G$ a set $F$ of
  edges so that $H=(V,E\DELTA F)$ is a cograph and $|F|$ is minimum. This
  NP-hard combinatorial optimization problem has recently found
  applications, e.g., in the context of phylogenetics. Efficient heuristics
  are hence of practical importance. The simple characterization of
  cographs in terms of their modular decomposition suggests that instead of
  editing $G$ one could operate directly on the modular decomposition. We
  show here that editing the induced $P_4$s is equivalent to resolving
  prime modules by means of a suitable defined merge operation on the
  submodules. Moreover, we characterize so-called module-preserving edit
  sets and demonstrate that optimal pairwise sequences of module-preserving
  edit sets exist for every non-cograph. This eventually leads to an exact
  algorithm for the cograph editing problem as well as  
	fixed-parameter tractable (FPT) results when cograph editing is 
	parameterized by the so-called modular-width. In addition, we provide two
  heuristics with time complexity $O(|V|^3)$, resp., $O(|V|^2)$.
}

\smallskip
\noindent
\keywords{Cograph Editing, Modular Decomposition, Module Merge,
  Prime Modules,  $P_4$}

\sloppy
\section{Introduction}

Cographs are of particular interest in computer science because many
combinatorial optimization problems that are NP-complete for arbitrary
graphs become polynomial-time solvable on cographs
\cite{Corneil:85,BLS:99,Gao:13}. This makes them an attractive starting
point for constructing heuristics that are exact on cographs and yield
approximate solutions on other graphs. In this context it is of
considerable practical interest to determine ``how close'' an input graph
is to a cograph.

An independent motivation recently arose in biology, more precisely in
molecular phylogenetics
\cite{HWL+15,lafond2015orthology,DEML:16,LDEM:16,NMMWH:17,GAS+:17}. In
particular, \emph{orthology}, a key concept in evolutionary biology in
phylogenetics, is intimately tied to cographs \cite{HWL+15}. Two genes in a
pair of related species are said to be orthologous if their last common
ancestor was a speciation event.  The orthology relation on a set of genes
forms a cograph \cite{HHH+13}, see \cite{HW:16} for a detailed discussion
and \cite{HSW:16,NMMWH:17,GAS+:17,Geiss:19a,Geiss:19x} for generalizations of these concepts.
This relation can be estimated directly from biological sequence data,
albeit in a necessarily noisy form. Correcting such an initial estimate to
the nearest cograph thus has recently become a computational problem of
considerable practical interest in computational biology
\cite{HWL+15}. However, the (decision version of the) problem to edit a
given graph with a minimum number of edits into a cograph is NP-complete
\cite{Liu:11,Liu:12,HW:15,HW:16a}.

As noted already in \cite{Corneil:81}, the input for several combinatorial
optimization problems, such as exam scheduling or several variants of
clustering problems, is naturally expected to have few induced paths on
four vertices ($P_4$s). Since graphs without an induced $P_4$ are exactly
the cographs, available cograph editing algorithms focus on efficiently
removing $P_4$s, see e.g.\
\cite{Liu:11,Liu:12,GPP:10,GHPP:13,DMZ:17,TWLS:17}.  The FPT-algorithm
introduced in \cite{Liu:11,Liu:12} takes as input a graph that is first
edited to a so-called $P_4$-sparse graph and then to a cograph. The basic
strategy is to destroy the $P_4$s in the subgraphs by branching into six
cases that eventually leads to an $O (4.612^k |V|^{9/2})$-time algorithm,
where $k$ is the number of required edits.  Algorithms that compute the
kernel of the (parameterized) cograph editing problem \cite{GPP:10,GHPP:13}
as well as the exact $O(3^{|V|}|V|)$-time algorithm \cite{TWLS:17} use the
modular-decomposition tree as a guide to locate the forbidden $P_4$s using
the fact that these are associated with prime modules. Nevertheless, the
basic operation in all of these algorithms is still the direct destruction
of the $P_4$s.

Cographs are recursively defined as follows: $K_1$ is a cograph, the
  disjoint union of cographs is a cograph, and the join of cographs is a
  cograph. This recursive definition associates a vertex labeled tree, the
  cotree, with each cograph, where a vertex label ``0'' denotes a disjoint
  union, while ``1'' indicates the join of the children is formed. It has
  been shown in \cite{Corneil:81} that each cograph has a unique cotree and
  conversely, every tree whose interior vertices are labeled alternatingly
  defined a unique cograph. A simple recognition algorithm starts with an
  input graph $G$. If $G$ is connected, then a node labeled ``1'' is
  inserted into the tree, the complement graph $\overline{G}$ is formed and
  the algorithm proceeds recursively on the connected components of
  $\overline{G}$. If $G$ is not connected, the tree-node is labeled ``0'',
  and the algorithm recurses on the components of $G$. If both $G$ and
  $\overline{G}$ are connected, then $G$ is not a cograph, and the
  algorithm terminates with a negative answer. A natural heuristic for
  cograph editing proceeds by finding a minimal cut in $G$ or
  $\overline{G}$, removes the cut-edges and proceeds with the modified
  graph. This idea is pursued in \cite{DEML:16,Dondi:17}.

A very different heuristic for cograph modification was recently
  proposed by Crespelle \cite{Crespelle:19x}. It corrects the neighborhood
  of each vertex separately. More precisely, an inclusion-minimal cograph
  editing $H_k$ of the induced subgraph $G_k:=G[\{x_1,\dots x_k\}]$ is
  computed from the correction $H_{i-1}$ of $G_{i-1}$ in such a way that
  only edges involving $x_i$ are inserted or deleted. It has the useful
  property that in each step the number of inserted or deleted edges is
  minimum, and it inserts or deletes no more than $|E(G)|$ edges in
  total. It is based on a general property of single-vertex augmentations
  in hereditary graph classes that are stable under the addition of
  universal vertices and isolated vertices, see e.g.\ \cite{Ohtsuki:81}. A
  key advantage is that it has linear time complexity, i.e., $O(|V|+|E|)$.

Cotrees are a special case of the much more general modular
  decomposition tree, which is well-defined for every graph and conveys
detailed information about its structure in a hierarchical manner
\cite{gallai-67}.  A subset $M\subseteq V$ is called a module of a graph
$G=(V,E)$, if all members of $M$ share the same neighbors in
$V\setminus M$.  A prime module is a module that is characterized by the
property that both, the induced subgraph $G[M]$ and its complement
$\overline{G[M]}$, are connected subgraphs of $G$. Cographs play a
particular role in this context as their modular decompositions are of a
special form: they are characterized by the absence of prime modules. In
particular, the cotree of a cograph coincides with its modular
decomposition tree \cite{gallai-67}. It is natural to ask, therefore,
  whether the modular decomposition tree can be manipulated in a such a way
  that all prime modules of a given graph are converted into ``series'' or
  ``parallel'' modules for which either $G[M]$ and or $\overline{G[M]}$ is
  disconnected. This is equivalent to converting $G$ into a cograph $G^*$.
  Every minimum edit set clearly is inclusion-minimal. However, not
    every minimum edit set -- and in particular not every inclusion-minimal
    edit set -- respects the module structure of $G$. Fig.\
    \ref{fig:example} below shows a pertinent example. In contrast to the
    editing approach of \cite{Crespelle:19x}, we pursue an approach that
  is modul-preserving in the sense that each module of $G$ is also a module
  of the edited graph $G^*$. We argue that this property is desirable in
  the context of orthology detection, because the corrected modular
  decomposition tree, i.e., the cotree of $G^*$ has a direct interpretation
  as event-labeled gene tree \cite{HHH+13,HWL+15}.

An alternative way of looking at the connection between cographs and
  their modular decomposition trees is to interpret the destruction of all
  $P_4$s in a cograph editing algorithm as the resolution of all prime
  modules in the edited graph $G^*$.  This simple observation suggests to
  edit the modules of $G$. The min-cut approach of \cite{DEML:16} is one
  possibility to achieve this. Here, we consider the merging of modules
  instead.  Every union $\bigcup_{i\in I} M_i$ of the connected components
$M_1$, \dots, $M_k$ of the edited graph $G^*[M]$ or $\overline{G^*[M]}$
forms a module $G^*$, while $\bigcup_{i\in I} M_i$ was not a module in the
graph $G$ before editing.  In this situation, we say that ``\emph{the
  modules $M_i$, $i\in I$ of $G$ are merged w.r.t.\ $G^*$}''.  Vertices
within a module $\bigcup_{i\in I} M_i$ share the same neighbors in
$V\setminus (\bigcup_{i\in I} M_i)$. It is sufficient therefore to adjust
the neighbors of certain submodules $M_i$ of $M$ to merge the $M_i$ in a
way that resolves the prime module $M$ to obtain $G^*$. In this setting, it
seems natural to edit the modular decomposition tree of a graph directly
with the aim of converting it step-by-step into the closest modular
decomposition tree of a cograph. To this end, one would like to break up
individual prime modules by means of the module merge operation.

The key results of this contribution are that (1) every prime node $M$ can
be resolved by a sequence of \emph{pairwise} merges of modules that are
children of $M$ in the modular decomposition tree, and (2) optimal cograph
editing can be expressed as optimal \emph{pairwise} module merging. To
prove these statements, we start with an overview of important properties
on cographs and the modular decomposition (Section \ref{sec:basic} and
\ref{sec:cograph}).  In Section \ref{sec:MM}, we then show that so-called
module-preserving edit sets are characterized by resolving any prime node
by module-merges. In particular, we show that any graph has an optimal edit
set that can be entirely expressed by merging modules that are children of
prime modules in the modular decomposition tree.  Finally in Section
\ref{sec:algo}, we summarize the results and show how they can
be used for establishing efficient heuristics for the cograph editing
problem.  We provide an exact algorithm that allows to optimally edit a
cograph via pairwise module-merges. As by-product, we obtain an 
FPT algorithm for the case that cograph editing is 
parameterized by the so-called modular-width \cite{GLO:13,ALM+17}.
We finish this paper with a short discussion on
how the latter method can be used to obtain a simple $O(|V|^2)$-time
heuristic.

\section{Basic Definitions}
\label{sec:basic}

We consider simple finite undirected graphs $G=(V,E)$ without loops.  The
complement $\overline G$ of a graph $G=(V,E)$ has vertex set $V$ and edge
set $E(\overline G)=\{xy\mid x,y\in V, x\neq y, xy\notin E\}$.  The
notation $G\DELTA F$ is used to denote the graph $(V,E\DELTA F)$, where
$\DELTA$ denotes the symmetric difference.
The disjoint union $G\cupdot H$ of two distinct graphs $G=(V,E)$ and
$H=(W,F)$ is simply the graph $(V\cupdot W, E\cupdot F)$.  The join $G
\oplus H$ of $G$ and $H$ is defined as the graph $(V\cupdot W, E\cupdot
F\cupdot \{xy\mid x\in V, y\in W\})$.  A graph $H=(W,F)$ is a
\emph{subgraph} of a graph $G=(V,E)$, in symbols $H\subseteq G$, if
$W\subseteq V$ and $F\subseteq E$.  If $H \subseteq G$ and $xy \in F$ if
and only if $xy\in E$ for all $x,y\in W$, then $H$ is called an
\emph{induced} subgraph.  We will often denote an induced subgraph
$H=(W,F)$ by $G[W]$.  A \emph{connected component} of $G$ is a connected
induced subgraph that is maximal w.r.t.\ inclusion.  We write $G\simeq H$
for two isomorphic graphs $G$ and $H$.

Let $G=(V,E)$ be a graph. The \emph{neighborhood} $N(v)$ of $v\in V$ is
defined as $N(v)=\{x\mid vx\in E\}$.  If there is a risk of confusion we
will write $N_G(v)$ to indicate that the respective neighborhood is taken
w.r.t.\ $G$. The \emph{degree} $\deg(v)$ of a vertex is defined as $\deg(v)
= |N(v)|$.
   
A \emph{tree} is a connected graph that does not contain cycles. A
\emph{path} is a tree where every vertex has degree $1$ or $2$. A
\emph{rooted} tree $T=(V,E)$ is a tree with one distinguished vertex
$\rho\in V$. We distinguish two further types of vertices in a tree: the
\emph{leaves} which are distinct from the root and are contained in only
one edge and the \emph{inner} vertices which are contained in at least two
edges.  The first inner vertex $\lca(x,y)$ that lies on both unique paths
from two vertices $x$, resp., $y$ to the root, is called \emph{lowest
  common ancestor} of $x$ and $y$.  We say that a rooted tree $T$
\emph{displays} the \emph{triple} $\rt{xy|z}$ if $x,y,$ and $z$ are leaves
of $T$ and the path from $x$ to $y$ does not intersect the path from $z$ to
the root of $T$.

It is well-known that there is a one-to-one correspondence between
(isomorphism classes of) rooted trees on $V$ and so-called hierarchies on
$V$.  For a finite set $V$, a \emph{hierarchy on $V$} is a subset $\mathcal
C$ of the power set $\mathcal P(V)$ such that $(i)$ $V\in \mathcal{C}$,
$(ii)$ $\{x\}\in \mathcal{C}$ for all $x\in V$ and $(iii)$ $p\cap q\in \{p,
q, \emptyset\}$ for all $p, q\in \mathcal{C}$.
\begin{theorem}[\cite{sem-ste-03a}]
  Let $\mc C$ be a collection of non-empty subsets of $V$.  Then, there is
  a rooted tree $T=(W,E)$ on $V$ with $\mc C = \{L(v)\mid v\in W\}$ if
  and only if $\mc C$ is a hierarchy on $V$.
  \label{A:thm:hierarchy}
\end{theorem}

\section{Cographs  and the Modular  Decomposition}
\label{sec:cograph}

\subsection{Introduction to Cographs}

\emph{Cographs} are defined as the class of graphs formed from a single
vertex under the closure of the operations of union and complementation,
namely: (i) a single-vertex graph $K_1$ is a cograph; (ii) the disjoint
union $G=(V_1\cupdot V_2,E_1\cupdot E_2)$ of cographs $G_1=(V_1,E_1)$ and
$G_2=(V_2,E_2)$ is a cograph; (iii) the complement $\overline{G}$ of a
cograph $G$ is a cograph.  
Condition (ii) can be replaced by the equivalent condition that the join 
$G_1 \oplus G_2$ is a cograph, since $G_1 \oplus G_2$ is the complement of 
$\overline{G}_1 \cupdot \overline{G}_2$.

The name cograph originates from
\emph{complement reducible graphs}, as by definition, cographs can be
``reduced'' by stepwise complementation of connected components to totally
disconnected graphs \cite{Seinsche1974}.

It is well-known that for each induced subgraph $H$ of a cograph $G$ either
$H$ is disconnected or its complement $\overline H$ is disconnected
\cite{BLS:99}.  This, in particular, allows representing the structure of a
cograph $G=(V,E)$ in an unambiguous way as a rooted tree $T=(W,F)$, called
\emph{cotree}: If the considered cograph is the single vertex graph $K_1$,
then output the tree $(\{u\}, \emptyset)$. Else if the given cograph $G$ is
connected, create an inner vertex $u$ in the cotree with label ``series'',
build the complement $\overline G$ and add the connected components of
$\overline G$ as children of $u$. If $G$ is not connected, then create an
inner vertex $u$ in the cotree with label ``parallel'' and add the
connected components of $G$ as children of $u$.  Proceed recursively on the
respective connected components that consists of more than one
vertex. Eventually, this cotree will have leaf-set $V\subseteq W$ and the inner
vertices $u\in W\setminus V$ are labeled with either ``parallel'' or
``series'' such that $xy\in E$ if and only if $u=\lca_T(x,y)$ is labeled
``series''. 

The complement of a path on four vertices $P_4$ is again a $P_4$ and hence,
such graphs are not cographs.  Intriguingly, cographs have indeed a quite
simple characterization as \emph{$P_4$-free} graphs, that is, no four
vertices induce a $P_4$.  A number of further equivalent characterizations
are given in \cite{BLS:99} and Theorem \ref{thm:cograph-characterize}.
Determining whether a graph is a cograph can be done in linear time
\cite{Corneil:85,BCHP:08}.

\subsection{Modules and the Modular Decomposition}

The concept of \emph{modular decompositions (MD)} is defined for arbitrary
graphs $G$ and allows us to present the structure of $G$ in the form of a
tree that generalizes the idea of cotrees. However, in general much more
information needs to be stored at the inner vertices of this tree if the
original graph has to be recovered.

The MD is based on the notion of modules. These are also known as
autonomous sets \cite{MR-84, Moh:85}, closed sets \cite{gallai-67}, clans
\cite{EGMS:94}, stable sets, clumps \cite{Blass:78} or externally related
sets \cite{HM-79}.  A \emph{module} of a given graph $G = (V,E)$ is a
subset $M\subseteq V$ with the property that for all vertices in $x,y\in M$
it holds that $N(y)\setminus M = N(x)\setminus M$. Therefore, the vertices
within a given module $M$ are not distinguishable by the part of their
neighborhoods that lie ``outside'' $M$.  We denote with $\M(G)$ the set of
all modules of $G=(V,E)$. Clearly, the vertex set $V$ and the singletons
$\{v\}$, $v\in V$ are modules, called \emph{trivial} modules. A graph $G$
is called \emph{prime} if it only contains trivial modules.  For a module
$M$ of $G$ and a vertex $v\in M$, we define the \out{M}-neighborhood of $v$
as $N(v)\setminus M$.  Since for any two vertices contained in $M$ the
\out{M}-neighborhoods are identical, we can equivalently define
$N(v)\setminus M$ as the \out{M}-neighborhood of the module $M$, where
$v\in M$.

We say that a module $M$ of $G$ is \emph{parallel}, resp.,
  \emph{series} if the induced subgraph $G[M]$, resp., the complement
  $\overline{G[M]}$ is disconnected.  If both $G[M]$ and $\overline{G[M]}$
  are connected, then $M$ is called \emph{prime}.

For a graph $G=(V,E)$ let $M$ and $M'$ be disjoint subsets of $V$. We say
that $M$ and $M'$ are adjacent (in $G$) if each vertex of $M$ is adjacent
to all vertices of $M'$; the sets are non-adjacent if none of the vertices
of $M$ is adjacent to a vertex of $M'$. Two disjoint modules are either
adjacent or non-adjacent \cite{Moh:85}. One can therefore define the
\emph{quotient graph} $G/P$ for an arbitrary subset $P\subseteq \M(G)$
of pairwise disjoint modules: $G/P$ has $P$ as its vertex set and
$M_iM_j\in E(G/P)$ if and only if $M_i$ and $M_j$ are adjacent in $G$.

A module $M$ is called \emph{strong} if for any module $M'\ne M$ either $M
\cap M' = \emptyset$, or $M \subseteq M'$, or $M' \subseteq M$, i.e., a
strong module does not \emph{overlap} any other module. The set of all strong
modules $\MD(G)\subseteq \M(G)$ thus forms a hierarchy, the so-called
\emph{modular decomposition} of $G$. While arbitrary modules of a graph
form a potentially exponential-sized family, the sub-family of
strong modules has size $O(|V(G)|)$ \cite{habib2004simple}.

Let $\P=\{M_1, \dots, M_k\}$ be a partition of the vertex set of a graph
$G=(V,E)$. If every $M_i\in \P$ is a module of $G$, then $\P$ is a
\emph{modular partition} of $G$. A non-trivial modular partition $\P=\{M_1,
\dots, M_k\}$ that contains only maximal (w.r.t\ inclusion) strong modules
is a \emph{maximal modular partition}.  We denote the (unique) maximal
modular partition of $G$ by $\Pmax(G)$. We will refer to the elements of
$\Pmax(G[M])$ as the \emph{the children of $M$}. This terminology is
motivated by the following considerations:

The hierarchical structure of $\MD(G)$ gives rise to a canonical tree
representation of $G$, which is usually called the \emph{modular
  decomposition tree} $\MDT(G)$ \cite{MR-84,HP:10}. The root of this tree
is the trivial module $V$ and its $|V|$ leaves are the trivial modules
$\{v\}$, $v \in V$. The set of leaves $L_v$ associated with the subtree
rooted at an inner vertex $v$ induces a strong module of $G$.  In
  other words, each inner vertex $v$ of $\MDT(G)$ \emph{represents} the
  strong module $L_v$.  An inner vertex $v$ is then labeled ``parallel'',
  ``series'', resp., ``prime'' if $L_v$ is a parallel, series, resp., prime
  module.  The strong module $L_v$ of the induced subgraph $G[L_v]$
associated to a vertex $v$ labeled ``prime'' is called prime module. Note,
the latter does not imply that the graph $G[L_v]$ is prime, however,
in all cases the quotient graph $G[L_v]/\Pmax(G[L_v])$ is prime \cite{HP:10}.  Similar to
cotrees it holds that $xy\in E$ if $u=\lca_{\MDT(G)}(xy)$ is labeled
``series'', and $xy\notin E$ if $u=\lca_{\MDT(G)}(xy)$ is labeled
``parallel''. However, to trace back the full structure of a given graph
$G$ from $\MDT(G)$ one has to store additionally the information of the
subgraph $G[L_v]/\Pmax(G[L_v])$ in the vertices $v$ labeled ``prime''.
Although, $\MD(G)\subseteq \M(G)$ does not represent all modules, we state
the following remarkable fact \cite{Moh:85,DGC:1997}: Any subset
$M\subseteq V$ is a module if and only if $M\in \MD(G)$ or $M$ is the union
of children of non-prime modules.  Thus, $\MDT(G)$ represents at least
implicitly all modules of $G$.

A simple polynomial time recursive algorithm to compute $\MDT(G)$ is as
follows \cite{HP:10}: (1) compute the maximal modular partition $\Pmax(G)$;
(2) label the root node according to the parallel, series or prime type of
$G$; (3) for each strong module $M$ of $\Pmax(G)$, compute $\MDT(G[M])$ and
attach it to the root node and proceed with $\Pmax(G[M])$.  The first
polynomial time algorithm to compute the modular decomposition is due to
Cowan \emph{et al.}\ \cite{CJS:72}, and it runs in $O(|V|^4)$.
Improvements are due to Habib and Maurer \cite{HM-79}, who proposed a cubic
time algorithm, and to M{\"u}ller and Spinrad \cite{MS:89}, who designed a
quadratic time algorithm. The first two linear time algorithms appeared
independently in 1994 \cite{CH:94, CS94}.  Since then a series of
simplified algorithms has been published, some running in linear time
\cite{DGC:01,CS:99,TCHP:08}, and others in almost linear time
\cite{DGC:01,CS:00,HPV:00, habib2004simple}.

For later reference we give the following lemma.

\begin{lemma}
  Let $M$ be a module of a graph $G=(V,E)$ and $M'\subseteq M$.  Then, $M'$
  is a module of $G[M]$ if and only if $M'$ is a module of $G$.  If $M$ is
  a strong module of $G$, then $M'$ is a strong module of $G[M]$ if and
  only if $M'$ is a strong module of $G$.  Moreover, if $M_1$ and $M_2$ are
  overlapping modules in $G$, then $M_1\setminus M_2$, $M_1\cap M_2$ and
  $M_1\cup M_2$ are also modules in $G$.
\label{lem:module-subg}
\end{lemma}
\begin{proof}
  The first and the last statement were shown in \cite{Moh:85}. We prove
  the second statement.

  Let $M\in \MD(G)$. Assume that $M'\subseteq M$ is a strong module of
  $G[M]$. Assume for contradiction that $M'$ is not a strong module of
  $G$.  Hence $M'$ must overlap some module $M''$ in $G$. This module $M''$
  cannot be entirely contained in $M$ as otherwise, $M''$ and $M'$ overlap
  in $G[M]$ implying that $M'$ is not a strong module of $G[M]$, a
  contradiction.  But then $M$ and $M''$ must overlap, contradicting that
  $M$ is strong in $G$.
  
  If $M'\subseteq M$ is a strong module of $G$ then it does not
    overlap \emph{any} module of $G$. Since every module of $G[M]$ is also
    a module of $G$, there cannot be a module of $G[M]$ that overlaps $M'$
    and thus, $M'$ is a strong module of $G[M]$.
\end{proof}

\subsection{Useful Properties of Modular Partitions}

First, we briefly summarize the relationship between cographs $G$ and the
modular decomposition $\MD(G)$. While the first three items are from 
\cite{BLS:99,Corneil:81}, the proof of the %fifth
fourth item can be found in 
\cite{Boeckner:98,HHH+13}.

\begin{theorem}[\cite{BLS:99,Corneil:81,HHH+13}]
Let $G=(V,E)$ be an arbitrary graph. Then the following statements are
equivalent. 
\begin{enumerate}
\item $G$ is a cograph.  
\item $G$ does not contain induced paths on four vertices.% $P_4$.
\item $\MDT(G)$ is the cotree of $G$ and hence, has no inner vertices
  labeled with ``prime''.
\item Define a set $\mc R(G)$ of triples as follows: For any three vertices
  $x,y,z\in V$ we add the triple $\rt{xy|z}$ to $\mc R(G)$ if either
  $xz,yz\in E$ and $xy\notin E$ or $xz,yz\notin E$ and $xy\in E$.  There is
  a tree $T$ that displays all triples in $\mc R(G)$.
\end{enumerate}
\label{thm:cograph-characterize}
\end{theorem}

For later explicit reference, we summarize in the next theorem several
results that we already implicitly referred to in the discussion above.

\begin{theorem}[\cite{GPP:10,HP:10,Moh:85}]
\label{thm:all}
The following statements are true for an arbitrary graph $G=(V,E)$:
\begin{itemize}
\item[(T1)]%HP:10
  The maximal modular partition $\Pmax(G)$ and the modular decomposition
  $\MD(G)$ of $G$ are unique.
\item[(T2)] %Former T5
  Let $\Pmax(G[M])$ be the maximal modular partition of $G[M]$, where $M$
  denotes a prime module of $G$ and $\P' \subsetneq \Pmax(G[M])$ be a
  proper subset of $\Pmax(G[M])$ with $|\P'|>1$.  Then,
  $\bigcup_{M'\in \P'}{M'} \notin \M(G)$.
\item[(T3)] %Former T6
  Any subset $M\subseteq V$ is a module of $G$ if and only if $M$ is either
  a strong module of $G$ or $M$ is the union of children of a non-prime
  module of $G$.
\end{itemize}
\end{theorem}

Statements (T1) and (T3) are clear.  Statement (T2) explains that none of
the unions of elements of a maximal modular partition of $G[M]$ are modules
of $G$, whenever $M$ is a prime module of $G$.  Moreover, Statement (T3)
can be used to show that all prime modules are strong.

\begin{lemma}
  Let $G=(V,E)$ be an arbitrary graph. Then, every prime module $M$ of $G$
  is strong.
\label{lem:strong-prime-modules}
\end{lemma}
\begin{proof}	
  Let $M$ be a prime module of $G$.
  Assume for contradiction that  $M$ is not strong in $G$.
  Thm.\ \ref{thm:all}(T3) implies that $M$ is the union of children of
  some non-prime module $M'$.
  Hence, there is a subset $\mathcal M \subsetneq \Pmax(G[M'])$ such that
  $M = \bigcup_{M'_i \in \mathcal M} M'_i$.
	  Note that $1<|\mc{M}|<|\Pmax(G[M'])|$, 
  since all $M'_i \in \Pmax(G[M'])$ are strong and
  $\bigcup_{M'_i \in \Pmax(G[M'])} M'_i = M'$ is non-prime. 
  As $M'$ is non-prime, it is either parallel or series. Since
  $M$ is a non-trivial union of elements in $\Pmax(G[M'])$,
  $G[M]$ is either disconnected (if $M'$ is parallel) or its complement
  $\overline{G[M]}$ is disconnected (if $M'$ is series).
  But then $M$ is non-prime; a contradiction.
  Thus, $M$ is a strong module of $G$.
\end{proof}

In what follows, whenever the term ``prime module'' is used it refers therefore
always to a strong module. 

\subsection{Cograph Editing}
\label{sec:CE}

Given an arbitrary graph we are interested in understanding how the graph can
be edited into a cograph. 
A well-studied problem is the following optimization problem. 

\begin{problem}[Optimal Cograph Editing]
  Given a graph $G=(V,E)$. Find a set $F\subseteq {{V}\choose{2}}$ of
  minimum cardinality such that $H=(V,E\DELTA F)$ is a cograph.
\end{problem} 

We will simply call an edit set of minimum cardinality an \emph{optimal
  (cograph) edit set}. For later reference we recall Lemma 9 of
  \cite{HWL+15}. It shows that it suffices to solve the cograph editing
	problem separately for each connected component of $G$.
\begin{lemma}[\cite{HWL+15}]
  Let $G=(V,E)$ be a graph with optimal edit set $F$.  Then $\{x,y\}\in
  F\setminus E$ implies that $x$ and $y$ are located in the same connected
  component of $G$.
\label{lem:CC}
\end{lemma}

Let $G=(V,E)$ be a graph and $F$ be an arbitrary edit set that transforms
$G$ to the cograph $H=(V, E\DELTA F)$.  If any module of $G$ is a module of
$H$, then $F$ is called \emph{module-preserving}.

\begin{proposition}[\cite{GPP:10}]
  Every graph has an optimal module-preserving cograph edit set.
\label{prop:FM}
\end{proposition}

The importance of module-preserving edit sets lies in the fact that
they update either all or none of the edges between
any two disjoint modules. It is worth noting that module preserving
edit sets do not necessarily preserve the property of modules being strong, 
i.e., although $M$ might be a strong module in $G$ it needs not to be
strong in $H$.

\begin{definition}
  Let $G=(V,E)$ be a graph, $F$ a cograph edit set for $G$ and $M$ be a
  non-trivial module of $G$.  The induced edit set in $G[M]$ is
	\[F[M]\coloneqq \{\{x,y\}\in F\mid x,y\in M\}.\]
\end{definition}

The next result shows that any optimal edit set $F$ can entirely expressed
by the union of edits within prime modules and that $F[M]$ is an optimal
edit set of $G[M]$ for any module $M$ of $G$.  Hence, if $F[M]$ is not
optimal for some module $M$ of $G$, then $F$ cannot be an optimal edit set
for $G$.

\begin{lemma}[\cite{GPP:10}]
  Let $G = (V,E)$ be an arbitrary graph and let $M$ be a non-trivial module
  of $G$.  If $F'$ is an optimal edit set of the induced subgraph $G[M]$
  and $F$ is an optimal edit set of $G$, then $(F\setminus F[M])\cup F'$
  is an optimal edit set of $G$. Thus, $|F[M]|=|F'|$. 

  Moreover, the optimal cograph editing problem can be solved independently
  on the prime modules of $G$.
\label{lem:edit-in-prime-modules}
\end{lemma}

\section{Module Merge is the Key to Cograph Editing} 
%Module Merge Deletes All $P_4$s} %as the key for cograph editing
\label{sec:MM}

Since cographs are characterized by the absence of induced $P_4$s, we can
interpret every optimal cograph-editing method as the removal of all
$P_4$s in the input graph with a minimum number of edits. A natural
strategy is therefore to detect $P_4$s and then to decide which %ones
edges must
be edited.  Optimal edit sets are not necessarily unique, see Figure
\ref{fig:example}.  The computational difficulty arises from the fact that
editing an edge of a $P_4$ can produce new $P_4$s in the updated
graph. Hence, we cannot expect \emph{a priori} that local properties of $G$
alone will allow us to identify optimal edits.

By Lemma \ref{lem:edit-in-prime-modules}, on the other hand, it is
sufficient to edit within the prime modules. Moreover, as shown in Figure
\ref{fig:example}, there are strong modules $\Ms$ in an optimal edited
cograph $H$ that are not modules in $G$.  Hence, instead of editing $P_4$s
in $G$, it might suffice to edit the \out{M_i}-neighborhoods for some
$M_i\in \Pmax(G[M])$ in such a way that they result in the new module $\Ms$
in $H$.  The following definitions are important for the concepts of the
``module merge process'' that we will extensively use in our approach.

\begin{definition}[Module Merge]
  Let $G$ and $H$ be arbitrary graphs with $V(H) \subseteq V(G)$ and
    let $\M(G)$ and $\M(H)$ denote
  their corresponding sets of all modules.  Consider a set
  $\mathcal{M}\coloneqq\{M_1,M_2,\dots,M_k\}\subseteq \M(G)$. We say that
  the modules in $\mathcal{M}$ are \emph{merged (w.r.t.\ $H$)}
  if 
  \begin{itemize} \setlength{\itemsep}{0pt}
  \item[{(i)}] $M_1,\dots,M_k \in \M(H)$,
  \item[{(ii)}] $M \coloneqq \bigcup_{i=1}^k M_i \in \M(H)$, and 
  \item[{(iii)}] $M\notin \M(G)$.
  \end{itemize}
  We use the symbols $\merge$ and $\to$ as operations that allows us
    to illustrate the merge process, that is, we write
    $M_1\merge\ldots\merge M_k = \merge_{i=1}^k M_i \to M$, whenever the
    modules $M_1,M_2,\dots,M_k$ are merged w.r.t.\ $H$ resulting in the
    module $M=\bigcup_{i=1}^k M_i$ of $H$. 
  \label{def:module-merge}
\end{definition}
The intuition is that the modules $M_1$ through $M_k$ of $G$ are merged
into a single new module $M$, their union, that is present in $H$ but not
in $G$. This, in particular, already defines all required edits to
  adjust the neighbors of the vertices in $\bigcup_{i=1}^k M_i$ in $G$
  resulting in the module $M=\bigcup_{i=1}^k M_i$ of $H$. It is easy to
verify that $\merge$ is commutative in the sense that if
$M_1\merge M_2 \to M$, then $M_2\merge M_1 \to M$. However, $\merge$ is
not necessarily associative.  To see this, consider the example in
Fig.\ \ref{fig:pp}. Although the module $\Ms_3$ in $H$ is obtained by
merging the modules $\{3\}$, $\{4\}$ and $\{5\}$, the set $\{3\}\cup\{4\}$
does not form a module in $H$. Hence, although
$\{3\}\merge\{4\}\merge\{5\}\to \Ms_3$, it does not hold that
$\{3\}\merge\{4\}\to \Ms$ for any module $\Ms$ in $H$.  Thus, we cannot
write $(\{3\}\merge\{4\})\merge \{5\} \to \Ms_3$.

It follows directly from Def.\ \ref{def:module-merge} that every new module
$M$ of $H$ that is not a module of $G$ can be obtained by merging trivial
modules: simply set $M = \bigcup_{x\in M} \{x\}$ and
$\merge_{x\in M} \{x\} \to M$ follows immediately. In what follows we will
show, however, that each strong module of $H$ that is not a module of $G$
can be obtained by merging the modules that are contained in $\Pmax(G[M])$
of some prime module $M$ of $G$.

\begin{figure}[tbp]
  \begin{center}
    \includegraphics[viewport=1 618 437 798, clip,scale=.8]{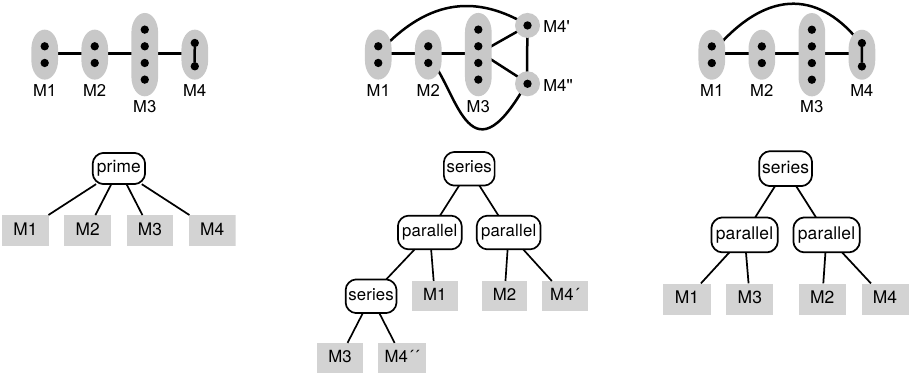}
    \caption{Shown are three graphs $G, H_1, H_2$ (from left to
      right). Maximal non-trivial strong modules are indicated by gray
      ovals in each graph and edges are used to show whether two modules
      are adjacent or not. The dots/lines within the modules are used to
      depict the vertices/edges within the modules. The modular
      decomposition trees up to a certain level are depicted below the
      respective graphs. This tree differs from the modular decomposition
      tree of the original graph $G, H_1$, and $H_2$, respectively, only
      from the unresolved leaf-nodes (gray boxes).  \newline \emph{Left:} A
      non-cograph $G$ is shown. The optimal edit set $F$ has cardinality
      $4$.  \emph{Center:} An optimal edited cograph $H_1=G\DELTA F$ is
      shown, where $F$ is not module-preserving.  None of the new strong
      modules of $H_1$ that are not modules of $G$ can be expressed as the
      union of the sets $M_1,\dots,M_4$.  Hence, none of these modules are
      the result of a module merge process.  \emph{Right:} An optimal
      edited cograph $H_2 = G\DELTA F$ is shown, where $F$ is
      module-preserving.  The new strong modules $M_1^{\star}, M_2^{\star}$
      of $H_2$ that are not modules of $G$ are two parallel
      modules. They can be written as $M_1^{\star} =M_1\cup M_3$ and
      $M_2^{\star}=M_2\cup M_4$. Hence, they are obtained by merging
      modules of $G$, in symbols: $M_1\protect\merge M_3\to M_1^{\star}$
      and $M_2\protect\merge M_4\to M_2^{\star}$. Here we have
      $F_{H_2}(M_1\protect\merge M_3\to M_1^{\star}) =
        F_{H_2}(M_2\protect\merge M_4\to M_2^{\star}) = F =\{\{x,y\}\mid
        x\in M_1, y\in M_4\}$.}
  \label{fig:example}
\end{center}
\end{figure}

When modules $M_1,\dots,M_k$ of $G$ are merged w.r.t.\ $H$ then all
vertices in $M = \bigcup_{h=1}^k M_h$ must have the same \out{M}-neighbors in
$H$, while at least two vertices $x\in M_i$, $y\in M_j$, $1\le i\neq j\le
k$ must have different \out{M}-neighbors in $G$.  Hence, in order to merge
these modules it is necessary to change the \out{M}-neighbors in $G$. 
However,  edit operations between vertices \emph{within} $M$ are dispensable for
obtaining the module $M$.
	
\begin{definition}[Module Merge Edit]
  Let $G=(V,E)$ be an arbitrary graph and $F$ be an arbitrary edit set
  resulting in the graph $H=(V,E\DELTA F)$. Let $H'\subseteq H$ be an
  induced subgraph of $H$ and suppose $M_1,\dots,M_k\in \M(G)$ are modules
  that have been merged w.r.t.\ $H'$ resulting in the module
  $M = \bigcup_{i=1}^k M_i \in \M(H')$. We then call
  \begin{equation}
    F_{H'}(\merge_{i=1}^k M_i \to M) \coloneqq \{\{x,v\}\in F\ \mid\ x\in M,	
    v\in V(H')\setminus M\} 
  \end{equation} 
  the \emph{module merge edits} associated with $\merge_{i=1}^k M_i \to M$ 
  w.r.t.\ $H'$. 
\end{definition}

By construction, the edit set $F_{H'}(\merge_{i=1}^k M_i \to M)$  comprises
exactly those (non)edges of $F$ that have been edited so that all vertices
in $M$ have the same \out{M}-neighborhood in $H'=(V',E')$. In particular, it
contains only (non)edges of $F$ that are not entirely contained in $G[M]$,
but entirely contained in $H'$. Moreover, (non)edges of $F$ that contain a
vertex in $V(H')$ and a vertex in $V\setminus V(H')$ are not considered as well.

Let $G$ be an arbitrary graph and $F$ be an optimal edit set that applied
to $G$ results in the cograph $H$. We will show that every optimal
module-preserving edit set $F$ can be expressed completely by means of
module merge edits. To this end, we will consider the prime modules
$M$ of the given graph $G$ (in particular certain children of $M$ that do
not share the same out-neighborhood) and adjust their out-neighbors to
obtain new modules. Illustrative examples
are given in Figure \ref{fig:example} and \ref{fig:pp}.

\begin{figure}[tbp]
  \begin{center}
   \includegraphics[viewport= 48 648 537 790, clip,scale=.7]{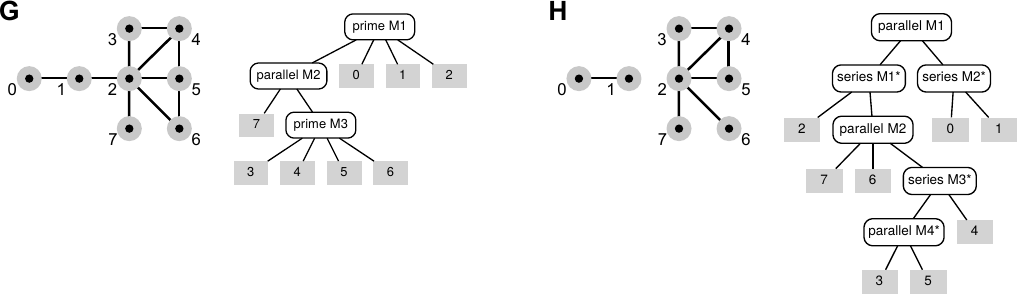}
   \caption{Illustration of the main results. Consider the non-cograph $G$,
     the cograph $H=G\DELTA F$ and the module-preserving edit set
     $F=\{\{1,2\}, \{5,6\}\}$.  The modular decomposition trees are
     depicted right to the respective graphs. \newline According to Theorem
     \ref{thm:Edit=Merge0}, both strong modules $M_1$ and $M_2$ of $H$ that
     are modules of $G$ are also strong modules of $G$ and correspond to
     the prime module $M_1$ and the parallel module $M_2$ in $G$,
     respectively.  Moreover, each of the new strong modules
     $\Ms_1,\dots, \Ms_4$ of $H$ are obtained by merging children of a
     prime module of $G$.  To be more precise, $\Ms_1$ and $\Ms_2$ are
     obtained by merging children of the prime module $M_1$ of $G$:
     $M_2\protect\merge\{2\}\to \Ms_1$ and
     $\{0\}\protect\merge\{1\}\to \Ms_2$ with
     $F_{H[M_1]}(M_2\protect\merge\{2\}\to \Ms_1) =
       F_{H[M_1]}(\{0\}\protect\merge\{1\}\to \Ms_2)=\{\{1,2\}\}$.  The
     new strong modules $\Ms_3$ and $\Ms_4$ are obtained by merging
     children of the prime module $M_3$ of $G$:
     $\{3\}\protect\merge\{5\}\to \Ms_4$ and
     $\{3\}\protect\merge\{4\}\protect\merge\{5\}\to \Ms_3$ with
     $F_{H[M_3]}(\{3\}\protect\merge\{5\}\to \Ms_4) =
       F_{H[M_3]}(\{3\}\protect\merge\{4\}\protect\merge\{5\}\to
       \Ms_3)=\{\{5,6\}\}$.  According to Cor.\ \ref{thm:Fmerge}, the set
     $F$ can be written as the union of the edit sets used to obtain the
     new merged modules of $H$.  \newline It is worth noting that not all
     strong modules of $G$ remain strong in $H$ (e.g.\ the prime module
     $M_3$) and that there are (non-strong) modules in $H$ (e.g.\ the
     module $\{6,7\}$) that are not obtained by merging children of prime
     modules of $G$.  }
  \label{fig:pp}
\end{center}
\end{figure}

We are now in the position to derive the main results,  Theorems
\ref{thm:Edit=Merge0}-\ref{thm:Fmerge}. We begin with showing that each
strong module of $H$ that is not a module of $G$ can be obtained
by merging some children of a particular chosen prime module of $G$. 
Moreover, we prove that any strong module of $H$ that is a module 
of $G$ must also be strong in $G$.

\begin{theorem}
  Let $G=(V,E)$ be an arbitrary graph, $F$ an optimal module-preserving
  cograph edit set, and $H=(V,E \DELTA F)$ the resulting cograph.  Then,
  each strong module $\Ms$ of $H$ is either a module in $G$ or there
    exists a prime module $P_{\Ms}$ of $G$ that contains $\Ms$ and is
    minimal w.r.t. inclusion, i.e., there is no prime module $P'_{\Ms}$ of
    $G$ with $\Ms \subseteq P'_{\Ms} \subsetneq P_{\Ms}$. In the latter
    case $\Ms$ is obtained by merging some modules in $\Pmax(G[P_{\Ms}])$.

Furthermore, if a strong module $\Ms$ of $H$ is a module in $G$, then $\Ms$
is a strong module of $G$.
\label{thm:Edit=Merge0}
\end{theorem}
\begin{proof}
Let $\Ms$ be an arbitrary strong module of $H$ that is not a module of $G$. We
show first that for the module $\Ms$ there is a prime module $P_{\Ms}$ of $G$
with $\Ms \subseteq P_{\Ms}$ such that there is no other prime module $P'_{\Ms}$
of $G$ with $\Ms \subseteq P'_{\Ms} \subsetneq P_{\Ms}$.

Since ${\Ms}$ is a module of $H$ but not of $G$ there are vertices $x\in \Ms$
and $y \in V \setminus \Ms$ with $\{x,y\} \in F$. Now, let $P_{\Ms}$ be the strong
module of $G$ containing $x$ and $y$ that is minimal w.r.t.\ inclusion, that is,
there is no other strong module of $G$ that is properly contained in $P_{\Ms}$
and that contains $x$ and $y$. Thus $\{x,y\} \in F[P_{\Ms}]$. Lemma
\ref{lem:edit-in-prime-modules} implies that $F[P_{\Ms}]$ is an optimal edit set
of $G[P_{\Ms}]$. Since $P_{\Ms}$ is minimal w.r.t. inclusion it holds that $x$
and $y$ are from distinct children $M_x,M_y \in \Pmax(G[P_{\Ms}])$. We continue
to show that this strong module $P_{\Ms}$ is indeed prime. Assume for
contradiction, that $P_{\Ms}$ is a non-prime module of $G$. If $P_{\Ms}$ is
parallel, then editing $\{x,y\}$ would connect the two connected components
$M_x,M_y$ of $G[P_{\Ms}]$. Then, it follows by Lemma \ref{lem:CC} that
$F[P_{\Ms}]$ is not optimal; a contradiction. By similar arguments for the
complement $\overline{G[P_{\Ms}]}$ it can be shown that $P_{\Ms}$ cannot be a
series module. Thus $P_{\Ms}$ must be prime. Since $F$ is module-preserving,
$P_{\Ms}$ is module in $H$. Hence, $P_{\Ms}$ and $\Ms$ cannot overlap, since
$\Ms$ is strong in $H$. However, since $x \in P_{\Ms} \cap \Ms$ and $y \in
P_{\Ms}$ but $y \notin \Ms$ we have $\Ms \subseteq P_{\Ms}$. Finally, since
$P_{\Ms}$ is chosen to be minimal w.r.t.\ inclusion, there exists in particular
no prime module $P'_{\Ms}$ of $G$ with $\Ms \subseteq P'_{\Ms} \subsetneq
P_{\Ms}$.

We continue to show that $\Ms$ is obtained by merging some child
  modules of $P_{\Ms}$ in $G$, say  $M_1,\dots,M_k\in\Pmax(G[P_{\Ms}])$.
Note that we just formally prove the existence of such a subset
  $\{M_1,\dots,M_k\} \subset \Pmax(G[P_{\Ms}])$ without explicitly
  constructing it.  To this end, we need to verify the three conditions of
Definition \ref{def:module-merge}, i.e., (i) $M_1,\dots,M_k \in \M(H)$,
(ii) $\Ms \coloneqq \bigcup_{i=1}^k M_i \in \M(H)$, and (iii)
$\Ms \notin \M(G)$.  Since each $M_i \in \Pmax(G[P_{\Ms}])$ is module of
$G$ and $F$ is module-preserving, Condition (i) is always
satisfied. Moreover, by assumption $\Ms \notin \M(G)$ and thus Condition
(iii) is satisfied.

It remains to show that Condition (ii) is satisfied.  To this
  end, we show that there are modules $M_1,\dots,M_k$ of $G$ (without
  explicitly constructing them) such that $\Ms = \bigcup_{i=1}^k M_i$. We
  prove this by showing that each module from $P_{\Ms}$ is either
  completely contained in, or disjoint from $\Ms$.  First, note that
$\Ms \neq P_{\Ms}$, since $\Ms$ is not a module of $G$. Second, $\Ms$
cannot overlap any $M_i \in \Pmax(G[P_{\Ms}])$, since $M_i$ is a module of
$H$ and $\Ms$ is strong in $H$. We continue to show that there is no
$M_i \in \Pmax(G[P_{\Ms}])$ such that $\Ms \subseteq M_i$. Assume for
contradiction that there is a module $M_i \in \Pmax(G[P_{\Ms}])$ with
$\Ms \subseteq M_i$. Note that $M_i$ cannot be prime in $G$, as otherwise
$\Ms \subseteq M_i=P'_{\Ms} \subsetneq P_{\Ms}$, contradicting the
minimality of $P_{\Ms}$. Moreover, $\Ms$ cannot overlap any
$M^i_j \in \Pmax(G[M_i])$, since $\Ms$ is strong in $H$ and any $M^i_j$ is
a module of $H$, since $F$ is module-preserving. Furthermore, since $M_i$
is non-prime in $G$ for any subset
$\{M^i_1,\ldots,M^i_l\} \subsetneq \Pmax(G[M_i])$ it holds that the set
$M' = \bigcup_{j=1}^l M^i_j$ is a module of $G$ (cf. Theorem
\ref{thm:all}(T3)). Since $\Ms$ is no module of $G$ it cannot be a union of
elements in $\Pmax(G[M_i])$.  Note that this especially implies that
$\Ms \neq M_i$ and $\Ms \neq M^i_j$ for all $M^i_j \in \Pmax(G[M_i])$. Now
it follows, that $\Ms \subset M^i_j$ for some $M^i_j \in \Pmax(G[M_i])$.
Repeating the latter arguments and since $G$ is finite, there must be a
minimal set $M^a_b$ with
$\Ms \subset M^a_b\subset \dots \subset M^i_j\subset M_i$. Now we apply the
latter arguments again and obtain that $\Ms \subset M'\in \Pmax(G[M^a_b])$
which is not possible, since $M^a_b$ is chosen to be the minimal module
that contains $\Ms$.  Thus, there is no $M_i \in \Pmax(G[P_{\Ms}])$ such
that $\Ms \subseteq M_i$.

Now, since $\Ms \neq P_{\Ms}$, and $\Ms$ does not overlap any $M_i \in
\Pmax(G[P_{\Ms}])$, and there is no $M_i \in \Pmax(G[P_{\Ms}])$ such that $\Ms
\subseteq M_i$, there must be a set $\{M_1,\ldots,M_k\} \subsetneq
\Pmax(G[P_{\Ms}])$ such that $\Ms = \bigcup_{i=1}^k M_i$. Thus, Condition (ii) is
satisfied and therefore $\Ms$ is obtained by merging modules in
$\Pmax(G[P_{\Ms}])$.

Hence, any strong module of $H$ is either a module of $G$ or obtained by merging
the children of a prime module of $G$.

Finally, assume that there is a strong module $\Ms$ in $H$ that is a
module of $G$. Assume that $\Ms$ is not strong in $G$. Then there is a
module $M$ in $G$ that overlaps $\Ms$. Since $F$ is module-preserving,
$M$ is a module in $H$ and thus, $M$ overlaps $\Ms$ in $H$; a contradiction.
Thus, any strong module $\Ms$ of $H$ that is also a module of $G$
must be strong in $G$.
\end{proof}

Theorem \ref{thm:Edit=Merge0} allows us to give the 
following definitions that we will use in the subsequent part. 
\begin{definition}
  Let $G=(V,E)$ be an arbitrary graph, $F$ an optimal module-preserving
  cograph edit set, and $H=(V,E \DELTA F)$ the resulting cograph. Let
  $\Ms$ be a strong module of $H$ but no module of $G$.

  We denote by $P_{\Ms}$ the prime module of $G$ that contains $\Ms$
  and is minimal w.r.t. inclusion, i.e., there is no prime module
  $P'_{\Ms}$ of $G$ with $\Ms \subseteq P'_{\Ms} \subsetneq
  P_{\Ms}$. Furthermore, we denote by
  $\mc{C}(\Ms) \subset \Pmax(G[P_{\Ms}])$ the set of children of $P_{\Ms}$
  that satisfies $\bigcup_{M_i\in \mc{C}(\Ms)} M_i = \Ms$.
\label{def:merge-by-primes}
\end{definition}

The next result provides a characterization of module-preserving edit
sets by means of module merge of the children of prime modules. 

\begin{theorem}
  Let $G=(V,E)$ be an arbitrary graph,
  $F$ an optimal cograph edit set, and
  $H=(V,E \DELTA F)$ the resulting cograph.
  Then $F$ is module-preserving for $G$ if and
  only if each new strong module $\Ms$ of $H$ that is not a module of $G$ is
  obtained by  merging the modules in $\mc{C}(\Ms)\subset \Pmax(G[P_{\Ms}])$, 
  in symbols $\merge_{M_i\in \mc{C}(\Ms) } M_i \to \Ms$.
\label{thm:Edit=Merge}
 \end{theorem}
\begin{proof}	
  If $F$ is an optimal and module-preserving edit-set for $G$, we can apply
  Theorem \ref{thm:Edit=Merge0}.
	 
  For the converse, assume for contraposition that $F$ is
  not module-preserving. Then, there is a module $M_i$ in $G$ that is not a
  module in $H$. Hence, there is a vertex $z\in V\setminus M_i$ and two
  vertices $x,y\in M_i$ such that $xz\in E(H)$ and $yz\notin E(H)$ and
  thus, either $\{x,z\}\in F$ or $\{y,z\}\in F$. There are two cases,
  either $xy\in E(H)$ or $xy\notin E(H)$. Since $H$ is a cograph we can
  apply Theorem \ref{thm:cograph-characterize} and conclude that either
  $\rt{yz|x}\in \mc R(H)$ or $\rt{xz|y}\in \mc R(H)$. Assume that
  $\rt{xz|y}\in \mc R(H)$ and let $T$ be the cotree of $H$. Since $T$
  displays $\rt{xz|y}$, the strong module $\Ms$ of $H$ located at the
  $\lca_T(x,z)$ contains the vertices $x$ and $z$ but not $y$. Moreover,
  since there is an edit $\{x,z\}$ or $\{y,z\}$ in $F$ there is a strong
  prime module $P_{\Ms}$ in $G$ that contains $x,y,z$ and is minimal
  w.r.t. inclusion. Note, $M_i\neq P_{\Ms}$ since $x,y\in M_i$ and $z\not\in
  M_i$.  Moreover, since $M_i$ is a module in $G$, but none of the unions
  of the children of $P_{\Ms}$ is a module of $G$ 
  (cf.\ Theorem \ref{thm:all}(T3)),
  we can conclude that $M_i\subseteq M'$, where $M'$ is a child of $P_{\Ms}$ in
  $G$.  Since $P_{\Ms}$ is the minimal prime module that contains $x,y,z$ and
  there is an edit $\{x,z\}$ or $\{y,z\}$ in $F$, the vertex $z$ must be
  located in a module different from the module $M'$ that contains both $x$
  and $y$.  Thus, $z\notin M'$. Therefore, there is no module in $G$ that
  contains $x$ and $z$ but not $y$.  Thus, $\Ms$ is no module of $G$.
  Since there is no module in $G$ that contains $x$ and $z$ but not $y$,
  the set $\Ms$ cannot be written as the union of children of any strong
  prime module $P_{\Ms}$ and thus, $\Ms$ is not obtained by merging modules of
  $\Pmax(G[P_{\Ms}])$. The case $\rt{yz|x}\in \mc R(H)$ is shown analogously.
\end{proof}

Combining the latter results,  it can be shown that for every graph $G$
there is always an optimal edit set such that the resulting cograph $H$
contains all modules of $G$ and any newly created strong module $\Ms$ of $H$
is obtained by merging the respective modules in $\mc{C}(\Ms)$.

\begin{theorem}
  Any graph $G=(V,E)$ has an optimal edit-set $F$ such that each strong
  module $M^{\star}$ in $H=(V,E\DELTA F)$ that is not a module of $G$ is
  obtained by merging modules in $\Pmax(G[P_{\Ms}])$, where $P_{\Ms}$ is a
  prime module of $G$.
\label{thm:Edit=Merge2}
 \end{theorem}
\begin{proof}
  Proposition \ref{prop:FM} implies that any graph has a module-preserving
  optimal edit set. Hence, we can apply Theorem \ref{thm:Edit=Merge} to
  derive the statement.
\end{proof}

Finally, the following result shows that each module-preserving edit set can indeed
be derived by considering the module merge edits only.

\begin{theorem}
  Let $G=(V,E)$ be an arbitrary graph,
  $F$ an optimal module-preserving cograph edit set,
  $H=(V,E \DELTA F)$ the resulting cograph, and
  $\mc{M}$ the set of all strong modules of $H$ that are no modules of $G$.
  Then,
  \[ F = \bigcup_{\Ms\in \mc{M}}
     \left(F_{H[P_{\Ms}]}(\merge_{M_i\in \mc{C}(\Ms)} M_i \to \Ms)\right). 
  \] 
\label{thm:Fmerge}
\end{theorem}
\begin{proof}
  We set
  $\Fs = \bigcup_{\Ms\in \mc{M}} \left(F_{H[P_{\Ms}]}(\merge_{M_i\in
      \mc{C}(\Ms)} M_i \to \Ms)\right)$.  Clearly, it holds that
  $\Fs \subseteq F$.  It remains to show that, $F \subseteq \Fs$.  First,
  observe, that every edit $\{x,y\} \in F$ is between distinct children
  $M_x,M_y \in \Pmax(G[P_{\Ms}])$ of a prime module $P_{\Ms}$ of $G$.  To
  see this, let $P_{\Ms}$ be a strong module of $G$ such that $x$ and $y$
  are in distinct children $M_x,M_y \in \Pmax(G[P_{\Ms}])$ and assume for
  contradiction that $P_{\Ms}$ is non-prime in $G$.  Let
  $F' \coloneqq \bigcup_{M_i \in \Pmax(G[P_{\Ms}])} F[M_i]$.  Since
  $P_{\Ms}$ is non-prime in $G$ it follows that $F'$ is an edit set for
  $G[P_{\Ms}]$, that is, $G[P_{\Ms}]\Delta F'$ is a cograph.  But
  $|F'|<|F[P_{\Ms}]|$; contradicting Lemma \ref{lem:edit-in-prime-modules}.
  Thus, every edit $\{x,y\} \in F$ is between distinct children
  $M_x,M_y \in \Pmax(G[P_{\Ms}])$ of a prime module $P_{\Ms}$ of $G$.
  
  Assume that $\{x,y\} \in F$, but $\{x,y\} \notin \Fs$.  By the latter
  arguments, there is a prime module $P_{\Ms}$ of $G$ with $x \in M_x$ and
  $y \in M_y$ and $M_x,M_y \in \Pmax(G[P_{\Ms}])$.  Now let $M'_x$ be the
  strong module of $H$ that contains $x$ but not $y$ and that is maximal
  w.r.t.\ inclusion.  Since $F$ is module-preserving, $M_x$ is a module in
  $H$. Moreover, since $M'_x$ is a strong module of $H$, the modules $M'_x$
  and $M_x$ do not overlap in $H$. Therefore, either $M_x \subsetneq M'_x$
  or $M'_x \subseteq M_x$.  We show first that the case
  $M_x \subsetneq M'_x$ is not possible.  Assume for contradiction, that
  $M_x \subsetneq M'_x$. Thus, there is a vertex $z \in M'_x\setminus M_x$.
  Since $P_{\Ms}$ is prime in $G$ and $M_x \in \Pmax(G[P_{\Ms}])$, we can
  apply Theorem \ref{thm:all} (T2) and conclude that there is no other
  module than $M_x$ in $G$ that entirely contains $M_x$ but not $y$.  Since
  $M_x \subsetneq M'_x \subsetneq P_{\Ms}$ it follows that $M'_x$ is a new
  strong module of $H$ and therefore, by Theorem \ref{thm:Edit=Merge0},
  obtained by merging modules
  $M_1,\ldots,M_k \in \mc{C}(M'_x) \subsetneq \Pmax(G[P_{\Ms}])$.  But then
  $\{x,y\} \in F_{H[P_{\Ms}]}(\merge_{M_i\in \mc{C}(M'_x)} M_i \to M'_x)
  \subseteq \Fs$; contradicting that $\{x,y\} \notin \Fs$.  Hence,
  $M'_x \subseteq M_x$.  Similarly, $M'_y \subseteq M_y$ for the strong
  module $M'_y$ of $H$ that contains $y$ but not $x$ and that is maximal
  w.r.t.\ inclusion.
  
  Consider now the strong module $\Ms$ of $H$ that is identified with the
  lowest common ancestor of the modules $\{x\}$ and $\{y\}$ within the
  cotree of $H$.  Then, there are distinct children in $\Pmax(H[\Ms])$,
  containing $x$ and $y$, respectively. Since $M'_x$ is the strong module
  of $H$ that contains $x$ but not $y$ and that is maximal w.r.t.\
  inclusion, we have $M'_x \in \Pmax(H[\Ms])$.  Analogously,
  $M'_y \in \Pmax(H[\Ms])$.
  
  Both, $M_x$ as well as $M_y$ are modules in $H$ and $G$. 
  Since $F$ is module-preserving, either all or none of the 
  edges between $M_x$ and $M_y$ are edited. 
  Since $\{x,y\} \in F$ we have, therefore, 
  $\{x',y'\} \in F$ for all $x'\in M'_x \subseteq M_x$ and 
  $y'\in M'_y \subseteq M_y$. 
  Let $F'\coloneqq \{\{x',y'\} \mid x' \in M'_x, y' \in M'_y\}$.
  By the latter argument $F' \neq \emptyset$ and $F' \subseteq F$.

  Note, the subgraphs $H[M'_x]$ and $H[M'_y]$ are cographs.  Since $\Ms$ is
  either a parallel or a series module in $H$, we have either (i)
  $H[M'_x \cup M'_y] = H[M'_x] \cupdot H[M'_y]$ or (ii)
  $H[M'_x \cup M'_y] = H[M'_x] \oplus H[M'_y]$, respectively.  Since $F'$
  comprises the edits $\{x',y'\}$ between \emph{all} vertices $x'\in M'_x$
  and $y'\in M'_y$, the graph $H[M'_x \cup M'_y] \DELTA F'$ is in case (i)
  the graph $H[M'_x] \oplus H[M'_y]$ and in case (ii)
  $H[M'_x] \cupdot H[M'_y]$.  By definition, in both cases
  $H[M'_x \cup M'_y] \DELTA F'$ is a cograph.  Note that $F'$ did not
  change the \out{M'_x \cup M'_y}-neighborhood and thus, the graph
  $H[\Ms] \DELTA F' = G[\Ms] \DELTA (F[\Ms] \setminus F')$ is a cograph as
  well.  Since $\{x,y\} \in F'\cap F[\Ms]$ it holds that
  $|F[\Ms] \setminus F'| < |F[\Ms]|$.  But then, $F[\Ms]$ is not optimal,
  and therefore, by Lemma \ref{lem:edit-in-prime-modules} the set $F$ is
  not optimal; a contradiction.
    
  In summary, there exists no edit $\{x,y\} \in F$ with $\{x,y\} \notin \Fs$.
  Hence, $F \subseteq \Fs$ and the statement follows.
\end{proof}

From an algorithmic perspective, Theorem \ref{thm:Fmerge} implies that it is
sufficient to correctly determine the set of strong modules of a resulting cograph $H$
that are no modules of the given graph $G$. Afterwards,  
the module-preserving edit set $F$ is obtained by taking all the edits needed 
for the corresponding module merge operations.
On the other hand, by Theorem \ref{thm:Edit=Merge2} it is ensured
that such a closest cograph $H$ that contains all modules of $G$ always exists.

\section{Pairwise Module Merge and Algorithmic Issues}
\label{sec:algo}

So far, we have shown that for an arbitrary graph $G=(V,E)$ there is an
optimal module-preserving edit set $F$ that transforms $G$ into the cograph
$H = (V, E \DELTA F)$ (cf.\ Theorem \ref{thm:Edit=Merge2}).  Moreover, this
edit set $F$ can be expressed in terms of edits derived by module merge
operations on the strong modules of $H$ that are no modules of $G$ (cf.\
Theorem \ref{thm:Fmerge}).  In what follows, we show that there is an
explicit order in which these individual merge operations can be
consecutively applied to $G$ such that all intermediate edit-steps result
in graphs that contain all modules of $G$, and, moreover, all new strong
modules produced in this edit-step are preserved in any further step.  In
Section \ref{sub:pmm}, we show that an optimal edit set can always be
obtained by a series of ``ordered'' pairwise merge operations.  In Section
\ref{subsec:algo}, we show that the latter ``order''-condition can even be
relaxed and that particular modules can be pairwisely merged in an
arbitrary order to obtain an optimal edited graph.

The next Lemma shows that the number of edits in an optimal edit set $F$
can be expressed as the sum of individual edits based on the
$\merge$-operator to obtain the strong modules in a cograph $H=G\DELTA F$
that are no modules in $G$.

\begin{lemma}
  Let $G=(V,E)$ be a graph, $F$ an optimal module-preserving cograph
  edit-set, and $H=(V,E\DELTA F)$ the resulting cograph.  Let
  $\mc{M} = \{\Ms_1,\dots,\Ms_n\}$ be the set of all strong modules of $H$
  that are no modules of $G$ and assume that the elements in $\mc{M}$ are
  partially ordered w.r.t.\ inclusion, i.e., $\Ms_i\subseteq \Ms_j$ implies
  $i\le j$.

  Let $\Ms\in \mc{M}$.  We set
  $F_{\Ms} \coloneqq \{\{x,v\}\in F\ \mid\ x\in \Ms, v\in P_{\Ms}\setminus
  \Ms\}$, that is, the set $F_{\Ms}\subseteq F$ comprises all edits in $F$
  that are used to obtain the module $\Ms$ within $G[P_{\Ms}]$.

  Furthermore, we set  $\sigma_{\Ms_1} = F_{\Ms_1}$ and    
  $\sigma_{\Ms_i} = F_{\Ms_i} \setminus (\bigcup_{j=1}^{i-1} F_{\Ms_j})$,
  $2\leq i\leq n$. Then 
  \[F=\bigcupdot_{i=1}^n \sigma_{\Ms_i} 
    \text{ and, thus, } |F| = \sum_{i=1}^n |\sigma_{\Ms_i}|\,.\]

  Moreover, for each intermediate graph
  $G_j = G \DELTA \left(\bigcup_{i=1}^j \sigma_{\Ms_i}\right)$ and any
  $\Ms_i \in \mc{M}$ with $ i-1 \le j$ we have
  \[G_{j}[\Ms_i]=H[\Ms_i]\,.\]

  In each step $j$ the induced subgraphs $G_{j}[\Ms_i]$ are already
  cographs for all sets $\Ms_i$ with $ i-1 \le j$ and hence
  $F[\Ms_i]\setminus \bigcup_{k=1}^j \sigma_{\Ms_k} = \emptyset$, for all
  $i-1 \le j$.
  \label{lem:merge-order}
\end{lemma}
\begin{proof}
  By Theorem \ref{thm:Edit=Merge0}, for each $\Ms\in \mc{M}$ there is an
  inclusion-minimal prime module ${P}_{\Ms}$ in $G$ and a set of children
  $\mc{C}(\Ms) \subseteq \Pmax(G[{P}_{\Ms}])$ such that
  $\merge_{M_i\in \mc{C}(\Ms)} M_i \to \Ms$. Thus, ${P}_{\Ms}$ and
  $\mc{C}(\Ms)$ exists and $\mc{C}(\Ms)$ is not empty.
  
  Now, we show that $|F|$ can be expressed by the sum of the size of the
  edits in $\sigma_{\Ms_i}$ To this end, observe that by Theorem
  \ref{thm:Fmerge},
  $F = \bigcup_{\Ms\in \mc{M}} \left( F_{H[P_{\Ms}]}(\merge_{M_i\in
      \mc{C}(\Ms)} M_i \to \Ms)\right)$.  Thus,
  $F = \bigcup_{\Ms\in \mc{M}} F_{\Ms}$.  By construction of
  $\sigma_{\Ms_i}$ it holds first that
  $\bigcup_{i=1}^n \sigma_{\Ms_i} = \bigcup_{i=1}^n F_{\Ms_i}$ and second
  that $\sigma_{\Ms_i} \cap \sigma_{\Ms_j} = \emptyset$ for all $i \neq
  j$. Hence, $F = \bigcupdot_{i=1}^n \sigma_{\Ms_i}$ and thus,
  $|F| = \sum_{i=1}^n |\sigma_{\Ms_i}|$.

  By construction, $\mc{M}$ is partially ordered w.r.t.\ inclusion. We want
  to show that $G_{j}[\Ms_i] = H[\Ms_i]$ for all $i-1 \le j$. To this end,
  we show that
  $F[\Ms_i]\setminus \bigcup_{k=1}^{j} \sigma_{\Ms_k} = \emptyset$, in
  which case after each step $j$ there are no more edits left to modify an
  edge between vertices within $\Ms_i$.  We show first that the latter is
  satisfied for all $1 \le i \le n$ and a fixed $j=i-1$.  Assume for
  contradiction that
  $\{x,y\}\in F[\Ms_i]\setminus \bigcup_{k=1}^{i-1} \sigma_{\Ms_k}$ and
  thus, $x,y\in \Ms_i$.  Since $\{x,y\} \in F = \bigcup_{k=1}^n F_{\Ms_k}$,
  there must be a module $\Ms_{\ell}\in \mc{M}$ such that
  $\{x,y\}\in F_{\Ms_{\ell}}$. By construction, $F_{\Ms_{\ell}}$ contains
  only the edits that affect the \out{\Ms_{\ell}}-neighborhood.  Thus,
  w.l.o.g.\ we can assume that $x\in \Ms_{\ell}$ and $y\not\in
  \Ms_{\ell}$. Since $\Ms_{\ell}$ and $\Ms_i$ are strong modules, they do
  not overlap, and therefore, $\Ms_{\ell}\subsetneq \Ms_i$. However, since
  $\mc{M}$ is partially ordered, we can conclude that ${\ell} < i$ and
  therefore, $\{x,y\} \in \bigcup_{k=1}^{i-1} \sigma_{\Ms_k}$. Hence,
  $\{x,y\}\notin F[\Ms_i]\setminus \bigcup_{k=1}^{i-1} \sigma_{\Ms_k}$; a
  contradiction.  Thus,
  $F[\Ms_i]\setminus \bigcup_{k=1}^{i-1} \sigma_{\Ms_k} = \emptyset$ for
  all $1 \le i \le n$.  But then, clearly
  $F[\Ms_i]\setminus \bigcup_{k=1}^{j} \sigma_{\Ms_k} = \emptyset$ holds
  for any $j \ge i-1$.  Thus, $G_{j}[\Ms_i] = H[\Ms_i]$ for all
  $i-1 \le j$.
\end{proof}

The following Lemma shows that, given the explicit order
$\mc{M} = \{\Ms_1,\dots,\Ms_n\}$ from Lemma \ref{lem:merge-order}, in which
the edits are applied to the graph $G$, the intermediate graphs $G_i$
retain all modules of $G$ and also all new modules $\Ms_j$, $j\le i$.

\begin{lemma}
  Let $G=(V,E)$ be an arbitrary graph, $F$ an optimal module-preserving
  cograph edit set, and $H=(V,E \DELTA F)$ the resulting cograph.
  Moreover, let $\mc{M} = \{\Ms_1,\dots,\Ms_n\}$ be the partially ordered
  (w.r.t. inclusion) set of all strong modules of $H$ that are no modules
  of $G_0\coloneqq G$, and choose $\sigma_{\Ms_i}$, $F_{\Ms_i}$ and the
  intermediate graphs $G_i$, $1\leq i\leq n$ as in Lemma
  \ref{lem:merge-order}.
	
  Then, any module $M'$ of $G$ is a module of $G_i$ and the set $\Ms_j$ is
  a module of $G_i$ for $1\leq i\leq n$ and any $j \le i$.
  \label{lem:merge-order-2}
\end{lemma}
\begin{proof}
  First note that $\sigma_{\Ms_i}$ affects only modules that are entirely
  contained in $P_{\Ms_{i}}$ and only their out-neighbors within
  $P_{\Ms_{i}}$.  Moreover $\Ms_j\subseteq \Ms_i$ implies that
  $P_{\Ms_{j}}\subseteq P_{\Ms_{i}}$.  The partial ordering of the
  elements in $\mc{M}$ implies that $P_{\Ms_{i}}$ remains a module in
  $G_i$.

  Before we prove the main statement, we show first that the following
    statement is satisfied:
  \begin{description}
    \item[Claim 1:] \emph{For every $M'$ with
        $\Ms_{i} \subsetneq M' \subsetneq P_{\Ms_{i}}$ we have
        $M'\neq\Ms_j \in \mc{M}$, $j \le i$ and $M'$ cannot be a module of
        $G$.}

    Let $M'$ be an arbitrary set with
    $\Ms_{i} \subsetneq M' \subsetneq P_{\Ms_{i}}$.  By the partial order
    of the elements in $\mc{M}$ we immediately observe that
    $M'\neq\Ms_j \in \mc{M}$ for any $j \le i$.  Now assume for
      contradiction that $M'$ is a module of $G$.  Note, all elements in
    $\Pmax(G[P_{\Ms_{i}}])$ are strong modules of $G$, and thus, do not
    overlap the module $M'$. Moreover, since $P_{\Ms_{i}}$ is prime in $G$,
    we can apply Theorem \ref{thm:all}(T2) and conclude that the union of
    elements of any proper subset $\P' \subsetneq \Pmax(G[P_{\Ms_{i}}])$
    with $|\P'|>1$ is not a module of $G$.  Taken the latter arguments
    together and because $M' \subsetneq P_{\Ms_{i}}$, we have
    $M' \subseteq M_{\ell}\in\Pmax(G[P_{\Ms_{i}}])$ for some $\ell$.
    Hence, $\Ms_{i} \subsetneq M' \subseteq M_{\ell}$.  However, since
    $\Ms_{i}$ is the union of some children
    $\P' \subseteq \Pmax(G[P_{\Ms_{i}}])$ of $P_{\Ms_{i}}$ it follows that
    $M_{\ell} \subseteq \Ms_{i}$; a contradiction.  This proves Claim
      1. \smallskip
  \end{description}

  We continue with proving the main statement by induction over $i$.
  Since $G_0=G$, the statement is satisfied for $G_0$.  We continue to show
  that the statement is satisfied for $G_{i+1}$ under the assumption that
  it is satisfied for $G_i$.

  For further reference, we note that $P_{\Ms_{i+1}}$ is a module of $G_i$,
  since $P_{\Ms_{i+1}}$ is a module of $G$ and by induction assumption.
  Moreover, $P_{\Ms_{i+1}}$ remains a module of $G_{i+1}$, since
  $G_{i+1}=G_i\DELTA \sigma_{\Ms_{i+1}}$ and $\sigma_{\Ms_{i+1}}$ does not
  affect the \out{P_{\Ms_{i+1}}}-neighborhood.  Furthermore, $\Ms_{i+1}$ is
  a module of $H$ and thus, of $H[P_{\Ms_{i+1}}]$.  Since
  $\sigma_{\Ms_{i+1}}$ contains all such edits to adjust $\Ms_{i+1}$ to a
  module in $H[P_{\Ms_{i+1}}]$, we can conclude that $\Ms_{i+1}$ is a
  module in $G_{i+1}[P_{\Ms_{i+1}}]$.  Therefore, Lemma
  \ref{lem:module-subg} implies that $\Ms_{i+1}$ is a module of $G_{i+1}$.
	
  Now, let $M'$ be an arbitrary module of $G$.  We proceed to show that
  $M'$ is a module of $G_{i+1}$.  By induction assumption, each module $M'$
  of $G$ is a module of $G_i$.  Since $F$ is module-preserving, $M'$ is
  also a module of $H$.  Hence, $M'\in \M(G)\cap \M(G_i)\cap \M(H)$.
  Moreover, by Claim 1 the case
  $\Ms_{i+1} \subsetneq M' \subsetneq P_{\Ms_{i+1}}$ cannot occur for any
  module $M'$ of $G$.
	
  Note, the module $M'$ cannot overlap $P_{\Ms_{i+1}}$, since
  $P_{\Ms_{i+1}}$ is strong in $G$.  Hence, for $M'$ one of the following
  three cases can occur: either $P_{\Ms_{i+1}} \subseteq M'$,
  $P_{\Ms_{i+1}} \cap M' = \emptyset$, or $M' \subsetneq P_{\Ms_{i+1}}$.
  In the first two cases, $M'$ remains a module of $G_{i+1}$, since
  $\sigma_{\Ms_{i+1}}$ contains only edits between vertices within
  $P_{\Ms_{i+1}}$, and thus, the \out{M'}-neighborhood is not affected.
  Therefore, assume that $M' \subsetneq P_{\Ms_{i+1}}$.  The module $M'$
  cannot overlap $\Ms_{i+1}$, since $\Ms_{i+1}$ is strong in $H$.  As shown
  above, the case $\Ms_{i+1} \subsetneq M' \subsetneq P_{\Ms_{i+1}}$ cannot
  occur, and thus we have either (1) $M' \subseteq \Ms_{i+1}$, or (2)
  $\Ms_{i+1} \cap M' = \emptyset$.
  \begin{description}
  \item[\textnormal{Case (1)}] Since $\sigma_{\Ms_{i+1}}$ affects only the
    \out{\Ms_{i+1}}-neighborhood, there is no edit between vertices in $M'$
    and $\Ms_{i+1}\setminus M'$ and, moreover,
    $G_{i+1}[\Ms_{i+1}] = G_i[\Ms_{i+1}]$.  By assumption, $M'$ is a module
    of $G_i$.  Thus, $M'$ is a module in any induced subgraph of $G_i$ that
    contains $M'$ and hence, in particular in $G_i[\Ms_{i+1}]$.  Hence,
    $M'$ is a module of $G_{i+1}[\Ms_{i+1}]$.  Now, we can apply Lemma
    \ref{lem:module-subg} and conclude that $M'$ is also a module of
    $G_{i+1}$.

  \item[\textnormal{Case (2)}] Assume for contradiction that $M'$ is no
    module of $G_{i+1}$.  Thus, there must be an edge $xy \in E(G_{i+1})$,
    $x\in M', y \in V \setminus M'$ such that for some other vertex
    $x'\in M'$ we have $x'y \notin E(G_{i+1})$.  Since $M'$ is a module of
    $G_i$ it must hold that $\{x,y\} \in \sigma_{\Ms_{i+1}}$ or
    $\{x',y\} \in \sigma_{\Ms_{i+1}}$.  Since $x,x'\notin \Ms_{i+1}$ and
    each edit in $\sigma_{\Ms_{i+1}}$ affects a vertex within $\Ms_{i+1}$,
    we can conclude that $y \in \Ms_{i+1}$.  Now, by construction of
    $F_{\Ms_{i+1}}$ and since $M' \subsetneq P_{\Ms_{i+1}}$, all edits
    between vertices of $\Ms_{i+1}$ and $M'$ are entirely contained in
    $F_{\Ms_{i+1}}$.  But this implies that none of the sets
    $\sigma_{\Ms_{\ell}}$ with $\ell>i+1$ contains $\{x,y\}$ or $\{x',y\}$.
    Hence, it holds that $xy \in E(H)$ and $x'y \notin E(H)$, which implies
    that $M'$ is no module of $H$; a contradiction.
  \end{description}
  \noindent
  Therefore, each module $M'$ of $G$ is a module of $G_{i+1}$.  \smallskip

%%%%%%%%%%%%%%%%%%%%%%%%%%%%%%%%%%%%%555%%%%%%%%%%%%%%%%%%%%%%%%%%%%%%%%%%%%%555

  We proceed to show that $\Ms_j \in \mathcal{M}$ is a module of $G_{i+1}$
  for all $j\leq i+1$.  As we have already shown this for $j=i+1$, we
  proceed with $j< i+1$.  By induction assumption, each module $\Ms_j$ is a
  module of $G_i$ for all $j<i+1$.  Note, the module $\Ms_j$ cannot overlap
  $P_{\Ms_{i+1}}$, since $\Ms_j$ is strong in $H$ and $P_{\Ms_{i+1}}$ is a
  module of $H$, because $F$ is module-preserving.  Hence, for $\Ms_j$ one
  of the following three cases can occur: either
  $P_{\Ms_{i+1}} \subseteq \Ms_j$, $P_{\Ms_{i+1}} \cap \Ms_j = \emptyset$,
  or $\Ms_j \subsetneq P_{\Ms_{i+1}}$.  In the first two cases, $\Ms_j$
  remains a module of $G_{i+1}$, since $\sigma_{\Ms_{i+1}}$ contains only
  edits between vertices within $P_{\Ms_{i+1}}$, and thus, the
  \out{\Ms_j}-neighborhood is not affected.  Therefore, assume that
  $\Ms_j \subsetneq P_{\Ms_{i+1}}$.  The module $\Ms_j$ cannot overlap
  $\Ms_{i+1}$, since both are strong in $H$.  Due to the partial ordering
  of the elements in $\mc{M}$, the case $\Ms_{i+1} \subsetneq \Ms_j$ cannot
  occur. Hence there are two cases, either (A) $\Ms_j \subseteq \Ms_{i+1}$,
  or (B) $\Ms_{i+1} \cap \Ms_j = \emptyset$.
  \begin{description}
  \item[\textnormal{Case (A)}] Since $\sigma_{\Ms_{i+1}}$ affects only the
    \out{\Ms_{i+1}}-neighborhood, there is no edit between vertices in
    $\Ms_j$ and $\Ms_{i+1}\setminus \Ms_j$.  By analogous arguments as in
    Case (1), we can conclude that $\Ms_j$ remains a module of
    $G_{i+1}[\Ms_{i+1}]$.  Lemma \ref{lem:module-subg} implies that $\Ms_j$
    is also a module of $G_{i+1}$.

  \item[\textnormal{Case (B)}] Assume for contradiction that $\Ms_j$ is no
    module of $G_{i+1}$.  Thus, there must be an edge $xy \in E(G_{i+1})$,
    $x\in \Ms_j, y \in V \setminus \Ms_j$ such that for some other vertex
    $x'\in \Ms_j$ we have $x'y \notin E(G_{i+1})$.  Since $\Ms_j$ is a
    module of $G_i$ it must hold that $\{x,y\} \in \sigma_{\Ms_{i+1}}$ or
    $\{x',y\} \in \sigma_{\Ms_{i+1}}$. Now, we can argue analogously as in
    Case (2) and conclude that $xy \in E(H)$ and $x'y \notin E(H)$, which
    implies that $\Ms_j$ is no module of $H$; a contradiction.
  \end{description}
  \noindent
  Therefore, each module $\Ms_j$, $j\leq i+1$ is a module of $G_{i+1}$.
\end{proof}

The latter two Lemmata show that there exists an explicit order, in which
all new modules $\Ms_i$ of $H$ can be constructed such that whenever a
module $\Ms_i$ is produced step $i$ the induced subgraph $G_{i-1}[\Ms_i]$
is already a cograph and, moreover, is not edited any further in subsequent
steps.

\subsection{Pairwise Module-merge} \label{sub:pmm}
Regarding Lemma \ref{lem:merge-order}, each module $\Ms_i$ is created by
applying the remaining edits $\sigma_{\Ms_i} \subseteq F_{\Ms_i}$ of the module
merge $\merge_{M'\in \mc{C}(\Ms_i)} M'\to \Ms_i$ to the previous intermediate
graph $G_{i-1}$. Now, there might be linear many modules in $\mc{C}(\Ms_i)$
which have to be merged at once to create $\Ms_i$. However, from an algorithmic
point of view the module $\Ms_i$ is not known in advance. Hence, in each step,
for a given prime module $M$ of $G$ an editing algorithm has to choose one of
the exponentially many sets from the power set $\mathcal P(\Pmax{G[M]})$ to
determine which new module $\Ms_i$ have to be created. For an algorithmic
approach, however, it would be more convenient to only merge modules in a
pairwise manner, since then only quadratic many combinations of choosing two
elements of $\Pmax{G[M]}$ have to be considered in each step.

The aim of this section is to show that for each of the $n$ steps of
creating one of the new strong modules $\mc{M} = \{\Ms_1,\dots,\Ms_n\}$ of
$H$ it is possible to replace the merge operation
$\merge_{M'\in \mc{C}(\Ms_i)} M'\to \Ms_i$ with a series of pairwise merge
operations.

Before we can state this result 
we have to define the following partition of strong modules of a resulting
cograph $H$ that are no modules of a given graph $G$.

\begin{definition}
  Let $G=(V,E)$ be an arbitrary graph, $F$ a module-preserving cograph edit
  set, and $H=(V,E \DELTA F)$ the resulting cograph.  Moreover, let
  $\Ms \in \mc{M}$ be a strong module of $H$ that is no module of $G$ and
  consider the partitions
  $\Pmax(H[\Ms]) = \{\widetilde{M}_1, \ldots, \widetilde{M}_k\}$ and
  $\mc{C}(\Ms) = \{\widehat{M}_1,\dots, \widehat{M}_l\}$.  We define with
  $\mc{X}(\Ms)=\{M_0,\ldots,M_n\}$ the set of modules that contains the
  maximal (w.r.t.\ inclusion) modules of
  $\Pmax(H[\Ms_i]) \cup \mc{C}(\Ms_i)$ as follows
\begin{multline*}
  \mc{X}(\Ms) \coloneqq
  \{ \widetilde{M}_i \in \Pmax(H[\Ms]) \mid \exists \widehat{M}_j \in \mc{C}(\Ms) \text{ s.t. } \widehat{M}_j \subseteq \widetilde{M}_i \} \\
  \cup \{ \widehat{M}_j \in \mc{C}(\Ms) \mid \exists \widetilde{M}_i \in
  \Pmax(H[\Ms]) \text{ s.t. } \widetilde{M}_i \subseteq \widehat{M}_j \}.
\end{multline*}

Note that for technical reasons the index of the elements in $\mc{X}$
starts with 0.

Furthermore, assume that $\mc{M} = \{\Ms_1,\dots,\Ms_n\}$ is a partially
ordered (w.r.t. inclusion) set of all strong modules of $H$ that are no
modules of $G$.  For each $\Ms_i \in \mc{M}$ let
$\mc{X}(\Ms_i) =
\{M_{i,0},\ldots,M_{i,l_i}\}$ %as in Def. \ref{def:merge-partition}
and set $\Ms_i(j) = \bigcup_{k=0}^j M_{i,k}$ for all $1 \le i \le n$ and
$1 \le j \le l_i$.  Then, we denote with
\begin{equation*}
  \mc{N}(\mc{M})= \{\Ns_1=\Ms_1(1),\ldots,\Ns_m=\Ms_n(l_n)\}
\end{equation*}
the set of all such $\Ms_i(j)$.  In particular, we assume that
$\mc{N}(\mc{M})$ is ordered as follows: if $\Ns_k=\Ms_i(j)$ and
$\Ns_l=\Ms_{i'}(j')$, then $k<l$ if and only if either $i<i'$, or $i=i'$
and $j<j'$, i.e., within $\mc{N}(\mc{M})$ the elements $\Ms_i(j)$ are
ordered first w.r.t. $i$, and second w.r.t. $j$.
\label{def:merge-partition}
\end{definition}

Although, we have already shown by Theorem \ref{thm:Edit=Merge} that any
new strong module $\Ms \in \mc{M}$ of $H$ can be obtained by merging the
modules from $\mc{C}(\Ms)$, we will see in the following that $\Ms$ can
also be obtained by merging the modules form $\mc{X}(\Ms)$. In particular,
we will see that if all elements in $\mc{X}(\Ms)$ are already modules of
the intermediate graph $\Gs$, then we can use any order of the elements
within $\mc{X}(\Ms)$ and successively merge them in a pairwise manner to
construct $\Ms$. As a consequence of doing pairwise module merges we obtain
in each step an intermediate module $\Ns \in \mc{N}(\mc{M})$.

To see the intention to use the partition $\mc{X}(\Ms)$ instead of
$\mc{C}(\Ms)$ observe the following.  Due to the order of the elements in
$\mc{M}$, the modules $\Ms_1,\dots,\Ms_n$ are constructed from bottom to
top, i.e., when module $\Ms$ is processed then all child modules from
$\Pmax(H[\Ms])$ are already constructed.  So, instead of obtaining $\Ms$ by
merging $\mc{C}(\Ms)$ we can indeed obtain $\Ms$ also by merging
$\Pmax(H[\Ms])$.  However, it might be the case that a non-trivial subset
$\bigcup_{i\in I} \widetilde{M}_i = \widehat{M}_j$ for some $j$, e.g., if
$\widehat{M}_j$ is a (strong) prime module of $G$ but not a strong module
of $H$.  But also in this case, we have to assure that $\widehat{M}_j$
remains a module of $H$.  In particular, we do not want to destroy
$\widehat{M}_j$ by merging the elements from $\Pmax(H[\Ms])$ in the
incorrect order.  Thus, we choose $\widehat{M}_j \in \mc{X}(\Ms)$ and do
not include the individual $\widetilde{M}_i, i \in I$ into $\mc{X}(\Ms)$.

Before we can continue, we have to 
show that $\mc{X}(\Ms)$ as given in Definition \ref{def:merge-partition} is indeed a partition of $\Ms$.

\begin{proposition}
  Let $G=(V,E)$ be an arbitrary graph, $F$ a module-preserving cograph edit
  set, and $H=(V,E \DELTA F)$ the resulting cograph.  Moreover, let $\Ms$
  be a strong module of $H$ that is no module of $G$ and consider the
  partitions $\Pmax(H[\Ms]) = \{\widetilde{M}_1, \ldots, \widetilde{M}_k\}$
  and $\mc{C}(\Ms) = \{\widehat{M}_1,\dots, \widehat{M}_l\}$.  Then
  $\mc{X}(\Ms)$ is a partition of $\Ms$.  As a consequence, for each
  $M \in \mc{X}(\Ms)$ there are index sets $I \subseteq \{1,\ldots,k\}$ and
  $J \subseteq \{1,\ldots,l\}$ such that
  $M = \bigcup_{i \in I} \widetilde{M}_i$ and
  $M = \bigcup_{j \in J} \widehat{M}_j$.
	\label{prop:merge-partition}
\end{proposition}
\begin{proof}
  First note that all $\widetilde{M}_i \in \Pmax(H[\Ms])$ are strong
  modules of $H$.  Moreover, all $\widehat{M}_j \in \mc{C}(\Ms)$ are strong
  modules of $G$.  Since $F$ is module-preserving it follows that none of
  the elements $\widetilde{M}_i \in \Pmax(H[\Ms])$ overlap any
  $\widehat{M}_j \in \mc{C}(\Ms)$, and vice versa.  Hence, for each
  $\widetilde{M}_i \in \Pmax(H[\Ms])$ there are three distinct cases:
  Either $\widetilde{M}_i\subseteq \widehat{M}_j$, or
  $\widehat{M}_j \subsetneq \widetilde{M}_i$, or
  $\widetilde{M}_i \cap \widehat{M}_j = \emptyset$ for all
  $\widehat{M}_j \in \mc{C}(\Ms)$.  Now, since $\Pmax(H[\Ms])$ and
  $\mc{C}(\Ms)$ are partitions of $\Ms$ it follows for each $x \in \Ms$
  that $x$ is contained in exactly one $\widetilde{M}_i \in \Pmax(H[\Ms])$
  and exactly one $\widehat{M}_j \in \mc{C}(\Ms)$ and either
  $\widetilde{M}_i\subseteq \widehat{M}_j$ or
  $\widehat{M}_j \subsetneq \widetilde{M}_i$.  By construction of
  $\mc{X}(\Ms)$ then either
  $\widetilde{M}_i= \widehat{M}_j\in \mc{X}(\Ms)$; or
  $\widetilde{M}_i\in \mc{X}(\Ms)$ and $\widehat{M}_j \not\in \mc{X}(\Ms)$;
  or $\widetilde{M}_i\not\in \mc{X}(\Ms)$ and
  $\widehat{M}_j \in \mc{X}(\Ms)$.  Thus, $\mc{X}(\Ms)$ is a partition of
  $\Ms$.
\end{proof}

Using the partitions $\mc{X}(\Ms), \Ms \in \mc{M}$ we now show that there
is a sequence of pairwise module merge operations that construct the
intermediate modules $\Ns_j \in \mc{N}(\mc{M})$ while keeping all modules
from $G$ as well as all previous modules $\Ns_i, i<j$.

\begin{lemma}
  Let $G=(V,E)$ be an arbitrary graph, $F$ an optimal module-preserving
  cograph edit set, $H=(V,E \DELTA F)$ the resulting cograph and
  $\mc{M} = \{\Ms_1,\dots,\Ms_n\}$ be the partially ordered
  (w.r.t. inclusion) set of all strong modules of $H$ that are no modules
  of $G$.

  For each $\Ms_i \in \mc{M}$ let
  $\mc{X}(\Ms_i) = \{M_{i,0},\ldots,M_{i,l_i}\}$ and assume that
  $\mc{N}\coloneqq \mc{N}(\mc{M})= \{\Ns_1,\ldots,\Ns_m\}$.  Note, each
  $\Ns_{l}$ coincides with some $\Ms_i(j)=\bigcup_{k=0}^j M_{i,k}$.  We
  define $F_{\Ms_i(j)} \subseteq F$ as the set
  \begin{equation*}
    F_{\Ms_i(j)} \coloneqq \{\{x,v\} \in F \mid x \in \Ms_i(j), v \in
    P_{\Ms_i} \setminus \Ms_i(j) \}.
  \end{equation*}

  Furthermore, set $G'_0=G$ and for each $1 \le l \le m$ define
  $G'_l = G'_{l-1} \DELTA \theta_l$ with
  \begin{equation*}
    \theta_l = \begin{cases}
      \ \emptyset &\mbox{,\ if } \Ns_l \text{ is a module of } G'_{l-1}\\
      \ F_{\Ns_l} \setminus \bigcup_{k=1}^{l-1} \theta_k &\mbox{,\ otherwise.}
    \end{cases}
  \end{equation*}
  If $\Ns_l$ is no module of $G'_{l-1}$, then $\theta_l$ contains exactly
  those edits that affect the out-neighborhood of $\Ns_l=\Ms_i(j)$ within
  $G[P_{\Ms_i}]$ that have not been used so far.
	
  The following statements are true for the intermediate graphs $G'_l$,
  $1 \le l \le m$:
  \begin{enumerate}
  \item Any set $\Ns_k$ is a module of $G'_l$ for all $k \le l$.
  \item Any module $M'$ of $G$ is a module of $G'_l$, i.e.,
    $\bigcup_{k=1}^{l} \theta_k$ is module-preserving.
  \item Either $G'_{l-1} \simeq G'_{l}$, or there are two modules
    $M_1,M_2 \in G'_{l-1}$ such that $M_1 \merge M_2 \to \Ns_l$ is a
    pairwise module merge w.r.t. $G'_l$.
  \end{enumerate}
  \label{lem:merge-order-3}
\end{lemma}
\begin{proof}
  Before we start to prove the statements, we will first show
  \begin{description}
    \item[Claim 1:] \emph{For each $1 \le l \le m$ it holds that
        $\Ns_l$ is a module of $H$.}

    By construction $\Ns_l = \Ms_i(j) = \bigcup_{k=0}^{j} M_{i,k}$ for some
    $1 \le i \le n$ and $1 \le j \le l_i$ with $M_{i,k} \in \mc{X}(\Ms_i)$.
    Moreover, for each $M_{i,k}$ it holds either that
    $M_{i,k} \in \Pmax{H[\Ms_i]}$ or $M_{i,k}$ is a union of elements in
    $\Pmax{H[\Ms_i]}$.  Therefore, $\Ns_l$ is a union of elements in
    $\Pmax{H[\Ms_i]}$.  Since $\Ms_i$ is a strong non-prime module of $H$,
    Theorem \ref{thm:all}(T3) implies that each union of elements in
    $\Pmax{H[\Ms_i]}$ is a module of $H$ and therefore, $\Ns_l$ is a module
    of $H$, which proves Claim 1.
  \end{description}

  We proceed to prove Statements 1 and 2 for each intermediate graph $G'_l$
  by induction over $l$.  Since $G'_0=G$, the Statements 1 and 2 are
  satisfied for $G'_0$.  We continue to show that Statements 1 and 2 are
  satisfied for $G'_{l+1}$ under the assumption that they are satisfied for
  $G_l$.

  We start to prove Statement 1.  First assume that $\Ns_{l+1}$ is
  already a module of $G'_l$.  Then, by construction it holds that
  $\theta_{l+1} = \emptyset$ and therefore, $G'_{l} = G'_{l+1}$.  Now, by
  induction assumption, it holds that all modules of $G$ and all modules
  $\Ns_k \in \mc{N}$, $k \le l$ are modules of $G'_{l} = G'_{l+1}$. Hence,
  all modules $\Ns_k \in \mc{N}$, $k \le l+1$ are modules of
  $G'_{l+1}$. Hence, if $\Ns_{l+1}$ is already a module of $G'_l$,
    then Statement 1 is satisfied for $G'_{l+1}$.
	
  Now assume that $\Ns_{l+1}$ is not a module of $G'_l$.  For the proof of
  Statement 1, we show first
  \begin{description}
    \item[Claim 2:] \emph{$\Ns_{l+1}$ is a module of $G'_{l+1}$.}
	
    By construction it holds that $\Ns_{l+1}=\Ms_i(j)$ for some
    $1 \le i \le n$ and $1 \le j \le l_i$.  Note that $P_{\Ms_i}$ is a
    module of $G$ and therefore, by induction assumption it is a module of
    $G'_l$.  Since $\theta_{l+1} \subseteq F_{\Ms_i(j)}$ did only affect
    the \out{\Ms_i(j)}-neighborhood within the prime module $P_{\Ms_i}$ of
    $G$ it follows that $P_{\Ms_i}$ is a module of $G'_{l+1}$.  Moreover,
    it holds that $F_{\Ms_i(j)} \subseteq \bigcup_{k=1}^{l+1} \theta_k$.
    Note that $F_{\Ms_i(j)}$ contains all those edits that affect the
    \out{\Ms_i(j)}-neighborhood within the prime module $P_{\Ms_i}$ of $G$.
    Hence, for all $x \in \Ms_i(j)$ and all
    $y \in P_{\Ms_i} \setminus \Ms_i(j)$ it holds that $xy \in E(H)$ if and
    only if $xy \in E(G'_{l+1})$.  The latter arguments then imply that
    $\Ms_i(j)$ is a module of $G'_{l+1}$ and therefore, $\Ns_{l+1}$ is a
    module of $G'_{l+1}$. This proves Claim 2.
  \end{description}

  Now, we proceed with showing
  \begin{description}
    \item[Claim 3:] \emph{$\Ns_k$, $k \le l$ is a module of
        $G'_{l+1}$.}

    Let $\Ns_k = \Ms_{i'}(j')$ and $\Ns_{l+1}=\Ms_i(j)$.  By induction
    assumption it holds that $\Ns_k$ is a module of $G'_l$.  By the
    ordering of elements in $\mc{N}$ it holds that $i' \le i$ and by the
    ordering of elements in $\mc{M}$ it then follows that
    $P_{\Ms_{i'}} \subseteq P_{\Ms_i}$ or
    $P_{\Ms_{i'}} \cap P_{\Ms_i} = \emptyset$.
		
    If $P_{\Ms_{i'}} \cap P_{\Ms_i} = \emptyset$ then $\Ns_k$ is not
    affected by the edits in $\theta_{l+1}$ since they are all within
    $P_{\Ms_i}$ and thus, $\Ns_{k}$ remains a module of $G'_{l+1}$.

    Now consider the case $P_{\Ms_{i'}} \subseteq P_{\Ms_i}$.  For later
    reference, we show
    \begin{description}
      \item[Claim 3':] \emph{$\Ns_k \subseteq \Ns_{l+1}$ or
          $\Ns_{k} \cap \Ns_{l+1} = \emptyset$.}

      If $i'=i$, then $j' < j$ and by construction,
      $\Ms_{i'}(j') \subseteq \Ms_i(j)$ which implies that
      $\Ns_{k} \subseteq \Ns_{l+1}$.  Assume now that $i'<i$ and thus,
      $\Ns_k = \Ms_{i'}(j') \subseteq \Ms_{i'}$.  Since $\Ms_i$ and
      $\Ms_{i'}$ are strong modules of $H$ they cannot overlap.  Therefore,
      and due to the ordering of the elements in $\mc{M}$ it follows that
      either $\Ms_{i'} \subset \Ms_{i}$ or
      $\Ms_{i'} \cap \Ms_{i} = \emptyset$.  If
      $\Ms_{i'} \cap \Ms_{i} = \emptyset$, then
      $\Ns_{k} \cap \Ns_{l+1} = \emptyset$.  If $\Ms_{i'} \subset \Ms_{i}$,
      then there is a module $M' \in \Pmax(H[\Ms_i])$ such that
      $\Ms_{i'} \in M'$, since $\Ms_i$ and $\Ms_{i'}$ are strong modules of
      $H$.  Furthermore, the set $\Ms_i(j)$ is a union of elements in
      $\mc{X}(\Ms_i)$ and for each $M_{i,h} \in \mc{X}(\Ms_i)$ it holds
      that either $M_{i,h} \in \Pmax(H[\Ms_i])$ or $M_{i,h}$ is the union
      of elements in $\Pmax(H[\Ms_i])$.  Hence, it follows that either
      $M' \subseteq \Ms_i(j)$ or $M' \cap \Ms_i(j) = \emptyset$.  If
      $M' \cap \Ms_i(j) = \emptyset$, then
      $\Ms_{i'}(j') \cap \Ms_i(j) = \emptyset$ and hence,
      $\Ns_{k} \cap \Ns_{l+1} = \emptyset$.  If, on the other hand,
      $M' \subseteq \Ms_i(j)$, then $\Ms_{i'}(j') \subseteq \Ms_i(j)$ and
      thus, $\Ns_{k} \subseteq \Ns_{l+1}$.  Therefore, in all cases we have
      either $\Ns_{k} \subseteq \Ns_{l+1}$ or
      $\Ns_{k} \cap \Ns_{l+1} = \emptyset$, which proves Claim 3'.
    \end{description}

    By Claim 3', we are left with the following two cases.
    \begin{description}
    \item[\textnormal{Case $\Ns_{k} \subseteq \Ns_{l+1}$.}]  Since
      $\theta_{l+1}$ did not effect edges within $\Ns_{l+1}$ it holds that
      $G'_l[\Ns_{l+1}] \simeq G'_{l+1}[\Ns_{l+1}]$.  By induction
      assumption, $\Ns_{k}$ is a module of $G'_l$ and hence, of
      $G'_l[\Ns_{l+1}] = G'_l[\Ms_i(j)]$.  Thus, $\Ns_{k}$ is a module of
      $G'_{l+1}[\Ms_i(j)]$.  Now, since $\Ns_{l+1}$ is a module of
      $G'_{l+1}$ and by Lemma \ref{lem:module-subg} it follows that
      $\Ns_{k}$ is a module of $G'_{l+1}$.
    \item[\textnormal{Case $\Ns_{k} \cap \Ns_{l+1} = \emptyset$.}]
      Recall that $\Ns_{k}=\Ms_{i'}(j')$ and $\Ns_{l+1}=\Ms_i(j)$ by
        the fact that $i'\le i$.  Moreover, as shown in the proof of
        Claim 2, we have
      $F_{\Ms_i(j)} \subseteq \bigcup_{k=1}^{l+1} \theta_k$.  Therefore,
      for all $x \in \Ms_i(j)$ and all $y \in \Ms_{i'}(j')$ it holds that
      $xy \in E(H)$ if and only if $xy \in E(G'_{l+1})$.  Now let
      $y,y' \in \Ms_{i'}(j')$ and $x \not\in \setminus \Ms_{i'}(j')$.
      Since $\Ms_{i'}(j')$ is a module of $H$, $xy$ as well as $xy'$ are
      either both edges $H$ or both are non-edges in $H$.

      If $x \in \Ms_i(j)$, then there are no further edits
      $F \setminus F_{\Ms_i(j)}$ that may affect any of these edges, since
      $F_{\Ms_i(j)} \subseteq \bigcup_{k=1}^{l+1} \theta_k$.  Thus,
      $xy \in E(G'_{l+1})$ if and only if $xy' \in E(G'_{l+1})$.

      If $x \not\in \Ms_i(j)$, then $xy$ as well as $xy'$ are not affected
      by $\theta_{l+1}$.  Hence, $xy' \in E(G'_{l+1})$ if and only if
      $xy' \in E(G'_l)$.  By induction assumption, $\Ms_{i'}(j')$ is a
      module of $G'_l$ and hence, $xy \in E(G'_l)$ if and only if
      $xy' \in E(G'_l)$ and therefore, $xy \in E(G'_{l+1})$ if and only if
      $xy' \in E(G'_{l+1})$.  Hence, $\Ns_{k}=\Ms_{i'}(j')$ is a module of
      $G'_{l+1}$, which proves Claim 3.
    \end{description}
  \end{description}

  By Claim 1, 2 and 3, Statement 1 is satisfied for $G'_{l+1}$.  We
  continue to prove Statement 2 and assume that $M'$ is a module of $G$ and
  by induction assumption $M'$ is a module of $G'_l$.
	
  Again, let $\Ns_{l+1}=\Ms_i(j)$ and consider the module $P_{\Ms_i}$ of
  $G$.  Since $P_{\Ms_i}$ is strong in $G$, it cannot overlap $M'$.  Thus,
  either $M' \cap P_{\Ms_i} = \emptyset$, or $P_{\Ms_i} \subseteq M'$, or
  $M' \subset P_{\Ms_i}$.
	
  If $M' \cap P_{\Ms_i} = \emptyset$ or $P_{\Ms_i} \subseteq M'$ then $M'$
  is not affected by the edits in $\theta_{l+1}$ since they are all within
  $P_{\Ms_i}$ and thus, $M'$ remains a module of $G'_{l+1}$.
	
  Hence, we only have to consider the case $M' \subset P_{\Ms_i}$. We
    show
  \begin{description}
    \item[Claim 4:] \emph{Either $M' \subseteq \Ns_{l+1}$ or
        $M' \cap \Ns_{l+1} = \emptyset$.}

    Note again, that the set $\Ms_i(j)$ is a union of elements in
    $\mc{X}(\Ms_i)$ and for each $M_{i,h} \in \mc{X}(\Ms_i)$ it holds that
    either $M_{i,h} \in \Pmax(G[P_{\Ms_i}])$ or $M_{i,h}$ is the union of
    elements in $\Pmax(G[P_{\Ms_i}])$.  Hence, $\Ms_i(j)$ is a union of
    elements in $\Pmax(G[P_{\Ms_i}])$.  Theorem \ref{thm:all}(T2) implies
    that no union of elements in $\Pmax(G[P_{\Ms_i}])$ of the prime module
    $P_{\Ms_i}$ is a module of $G$ and thus, $\Ms_i(j)$ cannot be a proper
    subset of $M'$.  Therefore, either $M' \subseteq \Ms_i(j)$ or
    $M' \cap \Ms_i(j) = \emptyset$ or $M'$ and $\Ms_i(j)$ overlap.
    However, the latter case cannot occur, since then $M'$ would either
    overlap one of the strong modules in $\Pmax(G[P_{\Ms_i}])$ or be a
    union of elements in $\Pmax(G[P_{\Ms_i}])$.  Thus, in all cases either
    $M' \subseteq \Ns_{l+1}$ or $M' \cap \Ns_{l+1} = \emptyset$, which
      proves Claim 4.
  \end{description}

  Now the same argumentation that was used to show Statement 1 can be used
  to show Statement 2.  Thus, Statement 2 is satisfied for $G'_{l+1}$.
	
  Finally, we prove Statement 3. To this end, assume that
  $G'_{l} \not\simeq G'_{l+1}$ and that $\Ns_{l+1}$ is no module of $G'_l$.
  We show that there are modules $M_1,M_2 \in G'_l$ with
  $M_1 \merge M_2 \to \Ns_{l+1}$ being a pairwise module merge
  w.r.t. $G'_{l+1}$.  Clearly, Items (ii) and (iii) of
  Def. \ref{def:module-merge} are satisfied, since $\Ns_{l+1}$ is a module
  of $G'_{l+1}$ but no module of $G'_l$.  It remains to show that there are
  two modules $M_1,M_2 \in G'_l$ with $M_1 \cup M_2 = \Ns_{l+1}$ and
  $M_1,M_2 \in G'_{l+1}$, i.e., Item (i) of Def. \ref{def:module-merge} is
  satisfied.  Note, $\Ns_{l+1} = \Ms_i(j)$ for some $i$ and $j\ge
  1$. Assume first that $j=1$.  Then, $\Ms_i(1)=M_{i,0} \cup M_{i,1}$ with
  $M_{i,0},M_{i,1} \in \mc{X}(\Ms_i)$.  For each $M_{i,h}$ it holds that
  $M_{i,h} \in \Pmax(H[P_{\Ms_i}])$ or $M_{i,h} \in \Pmax(G[P_{\Ms_i}])$.
  If $M_{i,h} \in \Pmax(G[P_{\Ms_i}])$ then $M_{i,h}$ is a module of $G$
  and by Statement 2, a module of $G'_l$ and $G'_{l+1}$.  If $M_{i,h}$ is
  no module of $G$, then $M_{i,h} \in \Pmax(H[P_{\Ms_i}])$ is a new strong
  module of $H$.  Therefore, there exists a $k<i$ such that
  $M_{i,h}=\Ms_k$.  Since $\Ms_k = \Ms_k(l_k)$ and by the ordering of
  elements in $\mc{N}$ it holds that $\Ms_k(l_k) = \Ns_{k'}$ for some
  $k' \le l$.  Thus, by Statement 1, all $M_{i,h}$ and therefore, $M_{i,0}$
  and $M_{i,1}$ are modules of $G'_l$ and $G'_{l+1}$.
	
  Now, assume that $\Ns_{l+1} = \Ms_i(j)$ with $j>1$.  Then,
  $\Ms_i(j) = \Ms_i(j-1) \cup M_{i,j}$.  By the same argumentation as
  before, it holds that $M_{i,j}$ is a module of $G'_l$ and $G'_{l+1}$.
  Moreover, by Statement 1, $\Ms_i(j-1)=\Ns_{l}$ is a module of $G'_l$ and
  $G'_{l+1}$.
	
  Thus, there are modules $M_1,M_2$ of $G'_l$ and $G'_{l+1}$ with
  $M_1 \cup M_2 = \Ns_{l+1}$.  Moreover, since for all
  $\{x,y\} \in \theta_{l+1}$ it holds that either $x \in \Ns_{l+1}$ and
  $y \in P_{\Ms_i}\setminus \Ns_{l+1}$, or vice versa, it follows that
  there are no additional edits contained in $\theta_{l+1}$ besides the
  edits of the module merge $M_1 \merge M_2 \to \Ns_{l+1}$ that transforms
  $G'_l$ into $G'_{l+1}$.
\end{proof}

We are now in the position to derive the main result of this section that
shows that optimal pairwise module-merge is always possible.

\begin{theorem} [Pairwise Module-Merge]
  For an arbitrary graph $G=(V,E)$ and an optimal module-preserving cograph
  edit set $F$ with $H=(V,E \DELTA F)$ being the resulting cograph there
  exists a sequence of pairwise module merge operations that transforms $G$
  into $H$.
  \label{thm:pair-mod}
\end{theorem}

\begin{proof}
  Set $\mc{M} = \{\Ms_1, \ldots, \Ms_n\}$,
  $\mc{N} = \{\Ns_1, \ldots, \Ns_m\}$,
  $\mc{X}(\Ms_i) = \{M_{i,0}, \ldots, M_{i,l_i}\}$, as well as $\theta_k$
  and $G'_k$ for all $1 \le k \le m$ as in Lemma \ref{lem:merge-order-3}.
  Again, we set $G_0 \coloneqq G$ and $H' \coloneqq G_m$.  By Lemma
  \ref{lem:merge-order-3} for each $1 \le k \le m$ there is a pairwise
  module merge $M_1 \merge M_2 \to \Ns_k$ that transforms $G_{k-1}$ to
  $G_k$.  Thus, there exists a sequence of module merge operations that
  transforms $G$ to some graph $H'$.
	
  In what follows, we will show that $\bigcupdot_{k=1}^{m} \theta_k = F$
  and therefore $H' \simeq H$, from which we can conclude the statement.
  For simplicity, we put $F' \coloneqq \bigcup_{k=1}^{m} \theta_k$.

  We start with showing 
  \begin{description}
    \item[Claim 1:] \emph{$F' \subseteq F$.}

    Note first that by construction it holds that
    $\theta_k \cap \theta_l = \emptyset$ for all $k \neq l$ and therefore,
    $F'=\bigcup_{k=1}^{m} \theta_k = \bigcupdot_{k=1}^{m} \theta_k$.  By
    construction of $\theta$ it holds that $\theta_k \subseteq F$ for all
    $1 \le k \le m$.  Hence, $F' \subseteq F$.
  \end{description}
	
  Before we show that $F = F'$, we will prove
  \begin{description}
    \item[Claim 2:] \emph{All strong modules of $H$ are modules of
        $H'$.}

    Lemma \ref{lem:merge-order-3}(1) implies that all modules $M'$ of $G$
    are modules of $H'$.  Moreover, Lemma \ref{lem:merge-order-3}(2)
    implies that all $\Ns_k \in \mc{N}$ are modules of $H'$.  Since for all
    $\Ms_i \in \mc{M}$ it holds that $\Ms_i = \Ms_i(l_i) = \Ns_k$ for some
    $1 \le k \le m$, the set $\Ms_i$ is a module of $H'$.  Since each
    strong module of $H$ is either a module of $G$ or a new module
    $\Ms_i \in \mc{M}$, all strong modules of $H$ are modules of $H'$.
  \end{description}

  We continue to show
  \begin{description}
    \item[Claim 3:] \emph{$F' \subsetneq F$ is not possible.}

    By Claim 1, $F' \subseteq F$. Thus assume for contradiction that
      $F'\neq F$. Since $F$ is an optimal edit set and $F' \subsetneq F$
    it follows that $H'$ is not a cograph.  Thus, there exist a prime
    module $M$ in $H'$ that contains no other prime module.

    We will now show that $M$ is a module of $H$ and that all
    $M_i \in \Pmax(H[M])$ are modules of $H'$.  Therefore, consider the
    strong module $P_M$ of $H$ that entirely contains $M$ and that is
    minimal w.r.t. inclusion.  Since $P_M$ is strong in $H$ it is, by
      Claim 2, also a module of $H'$.  Moreover, each module
    $M_i \in \Pmax(H[P_M])$ is strong in $H$ and, again by Claim 2, a
    module of $H'$ as well.  If $P_M=M$, then $M$ is a module of $H$ and we
    are done.  Assume now that $M \subsetneq P_M$.  Note that since $M$ and
    all $M_i \in \Pmax(H[P_M])$ are modules of $H'$ and $M$ is strong in
    $H'$ it holds that $M$ does not overlap any $M_i \in \Pmax(H[P_M])$.
    Moreover, $M \not\subseteq M_i$ since otherwise $M_i$ would have been
    chosen instead of $P_M$.  Thus, $M = \bigcup_{i \in I} M_i$ is the
    union of some elements $M_i$ in $\Pmax(H[P_M])$.  Since $P_M$ is a
    non-prime module of $H$ it follows by Theorem \ref{thm:all}(T3) that
    $M$ is a module of $H$.  Since $H$ is a cograph, the children
    $M_i \in \Pmax(H[P_M])$ of the non-prime module $P_M$ are the connected
    components of either $H[P_M]$ (if $P_M$ is parallel) or its complement
    $\overline{H[P_M]}$ (if $P_M$ is series).  Since
    $M = \bigcup_{i \in I} M_i$ is the union of some elements in
    $\Pmax(H[P_M])$ and $H[M] \subseteq H[P_M]$, we can conclude that
    $H[M]$, resp.\, its complement $\overline{H[M]}$, has as its connected
    components $M_i$, $i \in I$.  Thus,
    $\Pmax(H[M]) \subset \Pmax(H[P_M])$.  Hence, all $M_i$, $i\in I$ are
    strong modules in $H$ and, by the discussion above, all $M_i$ are
    modules of $H'$.
	
    Since all $M_i \in \Pmax(H[M])$ are modules of $H'$ and all
    $M'_j \in \Pmax(H'[M])$ are strong in $H'$, it holds that no
    $M_i \in \Pmax(H[M])$ can overlap any $M'_j \in \Pmax(H'[M])$.
    Therefore, if $M_i \cap M'_j \neq \emptyset$ then either
    $M'_j \subsetneq M_i$ or $M_i \subseteq M'_j$ for any $i$ and $j$.  If
    $M'_j \subsetneq M_i$ then $M_i$ must be the union of some elements in
    $\Pmax(H'[M])$.  However, since $M$ is prime in $H'$ no union of
    elements in $\Pmax(H'[M])$, besides $M$ itself, is a module of $H'$
    (cf. Theorem \ref{thm:all}(T2)).  Thus, $M_i$ cannot be a module of
    $H'$; a contradiction.  Hence, $M_i \subseteq M'_j$ and therefore, each
    $M'_j$ is the union of some elements in $\Pmax(H[M])$.  Note that this
    holds for any $M'_j \in \Pmax(H'[M])$, i.e., there are distinct sets
    $I_1,\ldots,I_{|\Pmax(H'[M])|}$ with
    $I_j \subsetneq \{1,\ldots,|\Pmax(H[M])|\}$ such that
    $M'_j = \bigcup_{i \in I_j} M_i$.  Hence, all $M'_j$ are modules of
    $H$.
	
    Since, $M$ is prime in $H'$ and $M$ did not contain any other prime
    module, it holds that all $H'[M'_j]$ are cographs.  Moreover, since all
    $M'_j$ are modules in $H$ and $M$ is prime in $H'$ it holds that there
    are at least two distinct $M'_k, M'_l \in \Pmax(H'[M])$ with
    $xy \in E(H')$ if and only if $xy \not\in E(H)$.  Thus,
    $F'' = \{\{x,y\} \mid x \in M'_k, y \in M'_l \} \subseteq F$.  Now,
    since all $H'[M'_j]$ are cographs it holds that $H'[M'_k \cup M'_l]$ is
    a cograph.
	
    Now, consider the graph $H'' = G \DELTA F \setminus F''$, and in
    particular the subgraph $H''[M] = G[M] \DELTA F[M] \setminus F''$.
    Again, since all $H'[M'_j]$ with $M'_j \in \Pmax(H'[M])$ are cographs
    it holds that $H[M'_j] \simeq H'[M'_j] \simeq H''[M'_j]$.  By
    construction of $F''$ for the previously chosen $M'_k$ and $M'_l$ it
    holds that $H'[M'_k \cup M'_l] \simeq H''[M'_k \cup M'_l]$ as well as
    $H[M \setminus (M'_k \cup M'_l)] \simeq H''[M \setminus (M'_k \cup
    M'_l)]$ is a cograph.  Moreover, since for all $x \in M'_k \cup M'_l$
    and all $y \in M \setminus (M'_k \cup M'_l)$ we have $xy \in E(H)$ if
    and only if $xy \in E(H'')$ it holds that $H''[M]$ is a cograph as
    well.  Note that $F'' \subseteq F[M]$ and $F'' \neq \emptyset$ and
    therefore, $|F[M] \setminus F''| < |F[M]|$.  But then, since
    $F[M] \setminus F''$ is an edit set for $G[M]$ and by Lemma
    \ref{lem:edit-in-prime-modules} the set $F$ is not optimal; a
    contradiction.  Thus, $F'$ cannot be a proper subset of $F$, which
      proves Claim 3.
  \end{description}

  Claim 1 and 3 immediately imply that $F=F'$.  In particular, we
  have
  $F'=\bigcupdot_{i=1}^{n}\bigcupdot_{j=1}^{l_i} \theta'_{\Ms_i(j)} =
  \bigcupdot_{i=1}^{n}\bigcupdot_{j=1}^{l_i} \theta_{\Ms_i(j)}= F$.
\end{proof}

It can easily be seen by the latter results that each of the modules in
$\mc{N}(\mc{M}) = \{\Ns_1, \ldots, \Ns_m\}$ that is created by a pairwise
module merge is either already a module of $G$, or a union of elements from
$\Pmax(G[M])$ of some prime module $M$ of $G$.

\subsection{A modular-decomposition-based Heuristic for Cograph Editing}
\label{subsec:algo}

\begin{algorithm}[htbp]
\caption{Pairwise Module Merge}
\label{alg:MM}
\begin{algorithmic}[1]\footnotesize
\STATE \textbf{INPUT:} A graph $G=(V,E)$.
\STATE $\Gs \gets G$;
\STATE $\Fs\gets \emptyset$;
\STATE $\MD(G) \gets$ \texttt{compute-modular-decomposition($G$)}.
\STATE $P_1,\dots, P_m$ be the prime modules of $G$ that are partially
ordered w.r.t. inclusion, i.e., $P_i \subseteq P_j$
implies $i \le j$.
\FOR{$p=1,\dots,m$}
\STATE $\mc{P}_p \gets \Pmax(G[P_p])$
\WHILE{$\Gs[P_p]$ is not a cograph}
\STATE $M_i,M_j\gets$\texttt{get-module-pair($\mc{P}_p$)}.
\COMMENT{according to Theorem \ref{thm:pair-mod}}
\IF{$M_i \cup M_j$ is no module of $\Gs$}
\STATE $\theta\gets$ \texttt{get-module-pair-edit($M_i\merge M_j \to N
\textnormal{ w.r.t. } G[P_p]$)}
\COMMENT{according to $\theta_l$ in Lemma \ref{lem:merge-order-3}}
\STATE $\Gs \gets \Gs \Delta\ \theta$
\ENDIF
\STATE $\mc{P}_p \gets \mc{P}_p \setminus \{M_i,M_j\} \cup \{N\}$
\ENDWHILE
\ENDFOR
\STATE \textbf{OUTPUT:} $H = \Gs$;
\end{algorithmic}
\end{algorithm}

Although the (decision version of the) optimal cograph-editing problem is
NP-complete \cite{Liu:11,Liu:12}, it is fixed-parameter tractable (FPT)
\cite{Protti:09,Cai:96,Liu:12}.  However, the best-known run-time for an
FPT-algorithm is $\mc{O}(4.612^k + |V|^{4.5})$, where the parameter $k$
denotes the number of edits.  These results are of little use for practical
applications, because the parameter $k$ can become quite large.  An exact
algorithm that runs in $O(3^{|V|}|V|)$-time is introduced in
\cite{TWLS:17}. Moreover, approximation algorithms are described in
\cite{NSS:01,DMZ:17}.  In the following we provide an alternative exact
algorithm for the cograph-editing problem based on pairwise module-merge.
The virtue of this algorithm is that it can be adopted very easily to
design a cograph-editing heuristic.

\begin{figure}[tbp]
\begin{center}
\includegraphics[viewport= 65 613 514 782,clip,scale=.8]{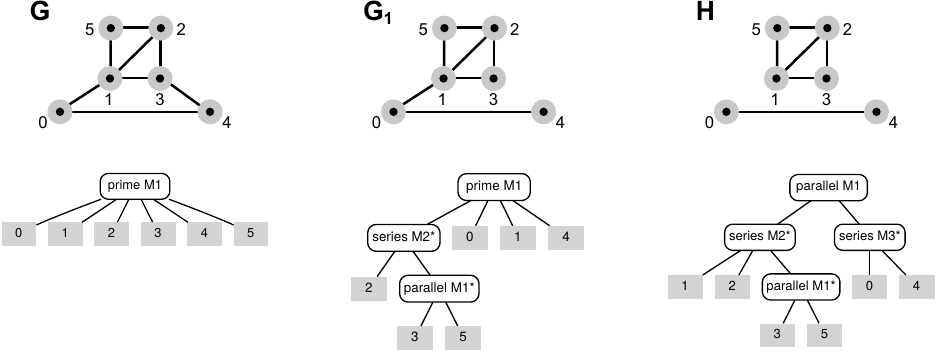}
\caption{
Illustration of Lemma \ref{lem:merge-order}-\ref{lem:merge-order-3},
Thm.\ \ref{thm:pair-mod} and the exact algorithm. Consider the
non-cograph $G$, the cograph $H=G\DELTA F$ and the optimal
module-preserving edit set $F=\{\{0,1\}, \{3,4\}\}$.
The modular decomposition trees are depicted below the respective graphs.
\newline
Let $\mathcal{M} =\{\Ms_1,\Ms_2,\Ms_3\}$ be the inclusion-ordered
set of strong modules of $H$ that are no modules of $G$.
For all modules $\Ms_i\in \mathcal{M}$ the inclusion-minimal module
$P_{\Ms_i}$ is the prime module
$M_1$ in $G$
\newline
In compliance with Lemma \ref{lem:merge-order-2} we start with
constructing the module $\Ms_1$.
By definition $F_{\Ms_1} = \{\{3,4\}\} = \sigma_{\Ms_1}$.
and we obtain $G_1 = G \DELTA \sigma_{\Ms_1}$.
Thus, $\{3\}\protect\merge\{5\}\to\Ms_1$ w.r.t.\ $G_1$.
Next, we continue with $\Ms_2$.
By construction, $F_{\Ms_2} = \{\{0,1\}, \{3,4\}\}$ and $\sigma_{\Ms_2}
= F_{\Ms_2} \setminus F_{\Ms_1} = \{\{0,1\}\}$.
We then obtain $G_2 = G_1 \DELTA \sigma_{\Ms_2} = H$.
Thus, $\protect\merge_{M_i \in \mc{C}(\Ms_2)} M_i \to \Ms_2$ w.r.t.
$G_2=H$.
The module $\Ms_3$ is now obtained for free, since $F_{\Ms_3} =
\{\{0,1\}, \{3,4\}\}$
and $\sigma_{\Ms_3} = F_{\Ms_3} \setminus (F_{\Ms_1} \cup F_{\Ms_2}) =
\emptyset$.
\newline
In compliance with Lemma \ref{lem:merge-order-3}, i.e., when
considering pairwise module merge only,
we start with constructing the module $\Ms_1(1)$.
Here, $\mc{X}(\Ms_1) = \{M_0 = \{3\}, M_1=\{5\}\}$ and
$\Ms_1(1)=\{3,5\}=\Ms_1$.
By definition, $F_{\Ms_1(1)} = \{\{3,4\}\} = \theta_{\Ms_1(1)}$
and we obtain $G_{1,1} = G_1 = G \DELTA \theta_{\Ms_1(1)}$.
Thus, $\{3\}\protect\merge\{5\}\to\Ms_1$ w.r.t.\ $G_{1,1} = G_1$.
Next, we continue with $\Ms_2(1)$ and $\Ms_2(2)$.
Here, $\mc{X}(\Ms_2) = \{M_0=\{1\},M_1=\{2\},M_2 = \Ms_1\}$ and
$\Ms_2(1)=\{1\}\cup\{2\}$
and $\Ms_2(2)=\{1,2,3,5\} = \Ms_2$.
By definition $\theta_{\Ms_2(1)} = F_{\Ms_2(1)} \setminus F_{\Ms_1(1)}
= \{\{0,1\}\}$ comprises the edits
to obtain the new module $\{1,2\}$.
Thus, $\{1\}\protect\merge\{2\}\to\Ms_2(1)$ w.r.t.\ $G_{2,1}$.
Then, since $F_{\Ms_2(2)} = F_{\Ms_2} = \{\{0,1\}, \{3,4\}\}$, we obtain
$\theta_{\Ms_2(2)} = F_{\Ms_2(2)} \setminus (F_{\Ms_1} \cup
\theta_{\Ms_2(1)} = \emptyset$.
Thus, there are no edits left to apply in order to derive at $H$, since
$G_{2,1}=G_{2,2}=G_2=H$.
Again, the module $\Ms_3$ is now obtained for free.
In all steps, we obtained the new modules by merging pairs of existing
modules.
}
\label{fig:algo}
\end{center}
\end{figure}

Algorithm \ref{alg:MM} contains two points at which the choice of a
particular module or a particular pair of modules affects performance and
efficiency. First, the function \texttt{get-module-pair()} returns
two modules of $\mc{P}$ in the correct order of the
sequence of pairwise module merge operations that transforms $G$ into $H$
(cf.\ Theorem \ref{thm:pair-mod}).
Second, subroutine \texttt{get-module-pair-edit()} is used to compute
the edits needed to merge the modules $M_i$ and $M_j$ to a new module
such that these edits affect only the vertices within $P_p$
(cf.\ Lemma \ref{lem:merge-order-3}).

\begin{lemma}
Let $\mathcal P(G)$ be the set of all strong prime modules of $G$ and 
suppose that Algorithm \ref{alg:MM} is applied on
the graph $G$ with $n=|V(G)|$.
If \texttt{get-module-pair()} is an ``oracle'' that always returns the
correct pair $M_i$ and $M_j$ and
\texttt{get-module-pair-edit()} returns the correct edit set $\theta$,
then Alg.\ \ref{alg:MM} computes an
optimally edited cograph $H$ in $O(m\Lambda h(G))\le O(n^2 h(G))$ time,
where $m$ denotes the number of strong prime modules in $G$ and 
$\Lambda=\max_{P \in \mc P(G)}|\Pmax(G[P])|$ is the size of the largest maximal strong partition among all prime modules $P\in \mc P(G)$, and
$h(G)$ is the maximal cost for evaluating \texttt{get-module-pair()}
and \texttt{get-module-pair-edit()}.
\label{lem:algoX}
\end{lemma}

\begin{proof}
The correctness of Algorithm \ref{alg:MM} follows directly from
Lemma \ref{lem:merge-order-3} and Theorem \ref{thm:pair-mod}.

The modular decomposition tree of a graph $G=(V,E)$ can be computed in
linear-time, i.e., $O(|V|+|E|) \le O(n^2)$ with $n=|V(G)|$, see \cite{CH:94,DGC:01,CS94,CS:99, TCHP:08}.  
It yields
  the partial order $P_1,\dots, P_m$ of the prime modules of $G$ (line 5)
  in time $O(n)$ by depth first search.  Then, we have to resolve
each of the $m$ prime modules and in each step in the worst case all
modules have to be merged stepwisely, resulting in an effort of
$O(|\Pmax(G[P_p])|)$ merging steps in each iteration.  Since $m\leq n$ and
$\Lambda\leq n$ we obtain $O(n^2 h(G))$ as an upper bound.
\end{proof}

In practice, the exact computation of the optimal editing requires
exponential effort. To be more precise, we show now the complexity
  $h(G)$ as in Lemma \ref{lem:algoX} using a naive brute-force method.
  Given a prime module $P$ with $\lambda=|\Pmax(G[P])|$ child modules there
  are $\lambda \choose 2$ possibilities for selecting the first module pair
  that has to be merged. After merging those two modules there are at most
  $\lambda-1$ modules left from which possibly two more have to be merged.
  In general in the $i$-th merging step there are at most
  $\lambda-i \choose 2$ possible merge pairs left.  This process have to
  repeat at most $(\lambda-4)$ times, since any module with less than four
  child modules cannot be prime.  In the worst case this adds up to
  $\prod_{i=4}^{\lambda} {i \choose 2} = \prod_{i=4}^{\lambda}
  {\frac{i!}{2!(i-2)!}} = \prod_{i=4}^{\lambda} {\frac{i \cdot (i-1)}{2}}$
  merge sequences per prime module of $G$ which gives $O((\lambda!)^2)$
  executions of \texttt{get-module-pair()} per prime module in $G$.
  Finding the optimal edit set for one merge operation of two modules
  $M_1,M_2 \in \Pmax(G[P])$ requires checking the $2^{\lambda-2}$
  combinations to add or remove edges to adjust the \out{M_1}- and
  \out{M_2}-neighbors w.r.t. to the remaining $\lambda-2$ modules.
  Therefore, for each of the remaining modules
  $M \in \Pmax(G[P])\setminus \{M_1,M_2\}$ there are either only edges or
  only non-edges between the vertices from $M$ and $M_1 \cup M_2$.
	In summary, for a given prime module $P$ the graph $G[P]$ can be optimally edited to a cograph in 
	$O((\lambda!)^22^{\lambda})$ time. Therefore, 
  with $\Lambda=\max_{P}|\Pmax(G[P])|$ being the size of the largest maximal strong partition among all prime modules $P$ of $G$,
  it follows that $h(G)\in O((\Lambda!)^22^{\Lambda})$. 
	We note in passing that $\Lambda$ is always less than or equal to the maximum degree in the modular decomposition tree, 
			which is also known as \emph{modular-width} \cite{GLO:13,ALM+17}. 
	Hence, the latter findings together with Lemma \ref{lem:algoX}  imply the following 
			\begin{Obs}
				The optimal cograph editing problem parameterized by the
				modular-width $k$ can be solved in $O((k!)^22^{k} |V|^2)$ time
				and thus, it is in FPT. 
			\end{Obs}

	Practical heuristics for
  \texttt{get-module-pair()} and \texttt{get-module-pair-edit()} can be
  implemented to run in polynomial time. In particular, as a main result,
  we can observe that it is always possible to find an optimal edit set by
  stepwisely merging only \emph{pairs} of modules.  
  Based on this, we
  provide in the following several strategies to improve the runtime of
  these heuristics.  

A simple greedy strategy yields a heuristic with $O(|V|^3)$ time complexity
as follows: In each call of \texttt{get-module-pair()} select the pair
$(M_i,M_j)$ in $\mc{P}$ where the edit set that adjusts the \out{M_i}- and
\out{M_j}-neighbors so that the \out{M_i \cup M_j}-neighborhood becomes
identical in $\Gs[P_p]$ has minimum cardinality.  This minimum edit set can
be obtained from \texttt{get-module-pair-edit()} by adjusting only the
out-neighbors of the smaller module to be identical to the out-neighbors of
the larger module.  The pseudocode for this heuristic is given in Algorithm
\ref{alg:heuristic} which is, in fact, a natural extension of the exact
Algorithm \ref{alg:MM}.  A detailed numerical evaluation will be discussed
elsewhere.

\begin{algorithm}[htbp]
\newcommand{\LINECOMMENT}[1]{\STATE \(\vartriangleright\)#1}
\caption{Pairwise Module Merge Heuristic}
\label{alg:heuristic}
\begin{algorithmic}[1]\footnotesize
\STATE \textbf{INPUT:} A graph $G=(V,E)$.
\STATE $\Gs \gets G$;
\STATE $\MD(G) \gets$ \texttt{compute-modular-decomposition($G$)}.
\STATE $P_1,\dots, P_m$ be the prime modules of $G$ that are partially
				ordered w.r.t. inclusion, i.e., $P_i\subseteq P_j$
				implies $i \le j$.
\STATE $A \gets$ zero initialized $|\MD(G)| \times |\MD(G)|$ matrix
\STATE $B \gets$ zero initialized $|\MD(G)| \times |\MD(G)| \times |\MD(G)|$ matrix
\LINECOMMENT Lines 8 to 15: Initialize $A$ where the entries $A_{ij}$ store the number $|V \setminus \{M_i \cup M_j\}|$  of vertices 
						 that need to be adjusted to merge the modules $M_i$ and $M_j$. 
             Initialize $B$ s.t.\ $B_{ijk}=1$ iff $M_i$ and $M_j$ have different out-neighborhoods w.r.t.\ $M_k$
\FOR{ each $\{M_i,M_j,M_k\} \in \binom{\MD(G)}{3}$ with $M_i,M_j,M_k$ being children of one and the same prime module $P$}
	\STATE $\textbf{if } $\out{M_i}$ \cap M_k \neq $\out{M_j}$ \cap M_k \textbf{ then } B_{ijk},B_{jik} \gets 1 \textbf{ end if}$
 	\STATE $\textbf{if } $\out{M_i}$ \cap M_j \neq $\out{M_k}$ \cap M_j \textbf{ then } B_{ikj},B_{kij} \gets 1 \textbf{ end if}$
 	\STATE $\textbf{if } $\out{M_j}$ \cap M_i \neq $\out{M_k}$ \cap M_i \textbf{ then } B_{jki},B_{kji} \gets 1 \textbf{ end if}$
	\STATE $A_{ij},A_{ji} \gets A_{ij} + |M_k| \cdot B_{ijk}$
 	\STATE $A_{ik},A_{ki} \gets A_{ik} + |M_j| \cdot B_{ikj}$
 	\STATE $A_{jk},A_{kj} \gets A_{jk} + |M_i| \cdot B_{jki}$
\ENDFOR
\FOR{$p=1,\dots,m$}
	\STATE $\mc{P} \gets \Pmax(G[P_p])$
	\WHILE{$|\mc{P}|>1$}
		\STATE $\theta \gets \emptyset$
		\COMMENT {$\theta$ denotes the set of (non)edges that will be edited}
		\STATE select two distinct modules $M_i$ and $M_j$ from $\mc{P}$ with $|M_i| \geq |M_j|$ that have a minimum value of $A_{ij}*|M_j|$.
		\LINECOMMENT Line 22 to 26: Compute the edits for adjusting the \out{M_i \cup M_j}-neighborhood s.t.\ $M_j$ has the same out-neighborhood as $M_i$
								 within $G[P_p]$. Note, since $P_p$ is a module of $G$, $M_j$ and $M_i$ have the same out-neighbors in $G$ after editing. 
		\IF{$A_{ij} \neq 0$, i.e., $M_i \cup M_j$ is no module of $\Gs$}
			\FOR {each $M_k \in \mc{P} \setminus \{M_i,M_j\}$}
				\STATE $\textbf{if } B_{ijk}=1 \textbf{ then } \theta \gets \theta \cup \{ xy \mid x \in M_j, y \in M_k\} \textbf{ end if}$
			\ENDFOR
		\ENDIF
		\LINECOMMENT Line 28 to 30: Adjust in $A$ the number of edits needed for merging the new module $M_i \cup M_j$ with some $M_k$
		\FOR{each $M_k \in \mc{P}\setminus \{M_i,M_j\}$}
			\STATE $A_{ik},A_{ki} \gets A_{ik}-|M_j| \cdot B_{ikj}$
		\ENDFOR
		\LINECOMMENT Line 32 to 34: Adjust in $A$ the number of edits needed for merging two modules $M_k$ and $M_l$
		\FOR{each $\{M_k,M_l\} \in \binom{\mc{P}\setminus \{M_i,M_j\}}{2}$}
			\STATE $A_{kl},A_{lk} \gets A_{kl}+|M_j| \cdot B_{kli}-|M_j| \cdot B_{klj}$
		\ENDFOR
		\STATE remove the $j$-th row and column $A$
		\STATE remove the $j$-th layer in all 3 dimensions of $B$
		\STATE in $\mc{P}$ replace $M_i$ with $M_i \cup M_j$
		\STATE $\mc{P}\gets \mc{P}\setminus \{M_j\}$
		\STATE $\Gs \gets \Gs \Delta\ \theta$
	\ENDWHILE
\ENDFOR
\STATE \textbf{OUTPUT:} $H = \Gs$;
\end{algorithmic}
\end{algorithm}

\begin{lemma}
  Algorithm \ref{alg:heuristic} outputs a cograph and has a time
  complexity of $O(|V|^3)$.
\end{lemma}

\begin{proof}
  First we show that Algorithm \ref{alg:heuristic} constructs a cograph.
  To this end we show that in each iteration of the main \emph{for}-loop
  (Lines 16 to 41) the corresponding prime module $P_p$ is edited such that
  the resulting subgraph $\Gs[P_p]$ is a cograph and $P_p$ is still a
  module of $\Gs$.

  Due to the processing order of the prime modules $P_1$, \dots, $P_m$
  constructed in Line 4, we may assume that, upon processing a prime module
  $P_p$, the induced subgraphs $\Gs[M], M \in \Pmax(G[P_p])$ are already
  cographs and all $M$ are modules of $\Gs$.  This holds in particular for
  the prime modules that do not contain any other prime module in the input
  graph $G$ and which, therefore, are processed first.  Hence, it suffices
  to show that if all $\Gs[M]$, $M \in \Pmax(G[P_p])$, are already cographs
  and all $M$ are modules in $\Gs$, then executing the $p-th$ iteration of
  the \emph{for}-loop results in an updated intermediate graph $G'$ with
  $G'[P_p]$ being a cograph and $P_p$ as well as all modules
  $M \in \Pmax(G[P_p])$ remain modules of $G'$.

  In Line 17, we define $\mc{P} = \Pmax(G[P_p])$ and therefore, by
  assumption, all $\Gs[M]$, $M \in \mc{P}$ are cographs and all $M$ are
  modules of $\Gs$.  In particular, the two sets $M_i$ and $M_j$ that are
  chosen first (in Line 20) are already cographs.  Moreover, since $M_i$
  and $M_j$ are modules of $\Gs$ if follows that $\Gs[M_i \cup M_j]$ is
  either the disjoint union $\Gs[M_i] \cupdot \Gs[M_j]$ or the join
  $\Gs[M_i] \oplus \Gs[M_j]$ of $\Gs[M_i]$ and $\Gs[M_j]$.  Thus,
  $\Gs[M_i \cup M_j]$ is already a cograph and none of the edges within
  $M_i \cup M_j$ is edited further.  It remains to show that applying the
  edits constructed in Line 24 result in the (new) merged module 
	$M_i \cup M_j$ of $\Gs \Delta \theta$. Note, if $M_i \cup M_j$ is already a module of
  $\Gs$ then Lines 22 to 26 are not executed and therefore,
  $\theta = \emptyset$, which implies that $M_i \cup M_j$ remains a module of
  $\Gs \Delta \theta$.  On the other hand, if $M_i \cup M_j$ is no module
  of $\Gs$ then the \emph{for}-loop in Lines 12 to 26 iterates over all
  modules $M_k$ in $\mc{P} \setminus \{M_i,M_j\}$ and adjusts the edges
  between $M_j$ and $M_k$ to be in accordance to the edges between $M_i$
  and $M_k$. Note that all those edits are within $P_p$. In particular, the
  \out{M_i \cup M_j}-neighborhood was adjusted only between vertices from $M_j$
  and vertices from $P_p \setminus (M_i \cup M_j)$.  After applying
  these edits, $M_i \cup M_j$ is therefore a module in
  $\Gs[P_p] \Delta \theta$. In particular, the \out{P_p}-neighborhood has not
  changed and $P_p$ is therefore a module of $\Gs$ as well as of
  $\Gs \Delta \theta$. Then, it follows by Lemma \ref{lem:module-subg}
  that $M_i \cup M_j$ is a module in $\Gs \Delta \theta$.
  To see that also all $M_k \in \mc{P} \setminus \{M_i,M_j\}$ remain modules in
  $\Gs \Delta \theta$ note first that $\mc{P}$ is a partition of $P_p$ and second,
  that only edges between $M_j$ and $M_k$ are edited for some $M_k \in \mc{P} \setminus \{M_i,M_j\}$.
  Moreover, if a (non)edge between $M_j$ and $M_k$ is edited, then all (non)edges
  $\{xy \mid x \in M_j, y \in M_k\}$ between $M_j$ and $M_k$ are edited.
  Thus all $M_k \in \mc{P} \setminus \{M_i,M_j\}$ remain modules of $\Gs[P_p] \Delta \theta$
  and therefore modules $\Gs \Delta \theta$.

  Now consider the prime module $P_{p+1}$ that is processed in the next
  iteration of the main \emph{for}-loop.  It can be easily seen that for
  $P_{p+1}$ we also have: 
	$\Gs[M], M \in \Pmax(G[P_{p+1}])$  is 
  a cograph and all $M$ are modules of $\Gs$, since all prime
  modules of $G$ that are subsets of $P_{p+1}$ are already processed, and
  therefore, are all those $M$ are non-prime modules of $\Gs$ and form
  cographs $\Gs[M]$.  Hence, by the same argumentation as before,
  $\Gs[P_{p+1}]$ is edited to a cograph by the next execution of the main
  \emph{for}-loop.  Thus, after processing all prime modules of $G$ the
  final graph $H$ is a cograph.

  Next, we show that Algorithm \ref{alg:heuristic} has a time complexity of
  $O(|V|^3)$.  Creating the modular decomposition in Line 3 can be done in
  linear time by the algorithms presented in, e.g.,
  \cite{DGC:01,CS:99,TCHP:08}.  Note that ``linear'' in this context means
	linear in the number of edges, i.e., 
  $O(|V|+|E|) \in O(|V|^2)$.  Initializing the matrices $A$ and $B$
  (Lines 8 to 15) requires time $O(|V|^3)$ since the corresponding
  \emph{for}-loop iterates over every ordered set of 3 strong modules of
  $G$ and there are at most $O(|V|)$ such modules.
  Moreover, checking if the out-neighborhoods of two modules $M_i$ and $M_j$
  w.r.t.\ a third module $M_k$ are identical 
  (the \emph{if}-statements in Lines 9 to 11)  can be done in
  constant time by checking the adjacencies between three arbitrary vertices,
  exactly 	one from each of the three modules. 
  For the remaining Lines 16 to 41 we can consider how often the inner \emph{while}-loop
  (Lines 18 to 40) is executed.  Therefore, note that within each execution always
  two modules are merged and there are $O(n)$ of those merge operations at
  most. This can most easily be seen by considering the matrix $A$ which has
  $\MD(G)$ rows and columns at first with $|\MD(G)| < |V|$.  Each row,
  respectively each column, of $A$ represents a module that is possibly
  selected for merging.  Moreover, within each iteration of the
  \emph{while}-loop, the matrix $A$ is reduced by one row, respectively
  one column.  This leads to no more than $|V|$ many executions of the
  \emph{while}-loop.
  Selecting the two modules $M_i$ and $M_j$ in Line 20
  requires $O(|V^2|)$ time.  Although, the \emph{for}-loop in Lines 23 to
  25 is executed $O(|V|)$ times and each partial edit set that is computed
  in Line 24 might contain more than $O(|V|)$ many edits, the whole edit
  set $\theta$ (constructed within Lines 23 to 25) contains no more
  than $O(|V|^2)$ edits.  Thus, executing Lines 12 to 26 requires
  $O(|V|^2)$ time at most.  Adjusting the matrix $A$ is done in two steps.
  Lines 28 to 30 iterates over $O(|V|)$ many modules $M_k$ and Lines 32 to
  34 iterates over $O(|V|^2)$ many pairs of modules $(M_k,M_l)$. Shrinking
  the matrices $A$ and $B$ in Lines 35 and 36 can technically be done in
  time $O(|V|)$ if we use a labeling function $l \colon \mathbb N \times \mathbb N$
  to index the values within the matrices,
  i.e., instead of reading $A_{ij}$ we read $A_{l(i),l(j)}$.
  Then we just have to relabel those indices,
  i.e., $l(x) \gets l(x)+1$ for all $x > j$.
  In that way we do not have to remove anything from $A$ or $B$.  
  Line 37 and 38 can also be done in $O(|V|)$ time and applying the edits in Line 39
  requires at most $O(|V|^2)$ time.  In summary, executing a single
  iteration of the main \emph{for}-loop requires $O(|V|^2)$ time, which
  yields a total time complexity of $O(|V|^3)$.
\end{proof}

The heuristic as given in Algorithm \ref{alg:heuristic} is deterministic and therefore lacks of a
randomization component which would be helpful in order to sample solutions and construct a consensus
cograph. However, randomization can be introduced easily by selecting a pair of modules $M_i$ and $M_j$
in line 20 with a probability inversely correlated with the value of $A_{ij} \cdot |M_j|$.
Moreover, with probability $p = |M_i| / (|M_i|+|M_j|)$ the edits $\{ xy \mid x \in M_j, y \in M_k\}$
can be selected in line 24 and otherwise $\{ xy \mid x \in M_i, y \in M_k\}$ with probability $1-p$.

An even simpler (but probably less accurate) heuristic with time complexity $O(|V|^2)$ can be obtained
by randomly selecting the next pair of modules $M_i$ and $M_j$ that have to be merged.
Such a procedure would not require the computation of the matrices $A$ and $B$ at all.
Nevertheless, this $O(|V|^2)$-time heuristic requires that computing the edit set $\theta$ can be done in $O(|V|)$ time.
However, this is possible if we only track the $O(|V|)$ many edits on the corresponding quotient graph $\Gs[P_p]/\Pmax(G[P_p])$
and recover the $O(|V|^2)$ many individual edits from that only once in a single post-processing step at the end.

Cograph editing heuristics based on the destruction of P4s requires $O(|V|^4)$ time merely for enumerating all P4s.
Thus, using module merges as editing operation may lead to significantly faster cograph editing heuristics.  

%\section*{Acknowledgements}

%We thank the anonymous referees for their helpful comments. 

% BibTeX users please use one of
%\bibliographystyle{spbasic}      % basic style, author-year citations
%\bibliographystyle{spmpsci}      % mathematics and physical sciences
\bibliographystyle{plain}       % APS-like style for physics
\bibliography{biblio}   % name your BibTeX data base

\end{document}